\documentclass[draftcls,11pt,onecolumn]{IEEEtran}
\usepackage{bbm}
\usepackage{mathbbold}
\usepackage{mathrsfs}
\usepackage{stmaryrd}
\usepackage{amssymb}
\usepackage{amsmath,epsfig}
\usepackage{multirow}
\usepackage{tabularx}
\usepackage{latexsym}
\usepackage{times}
\usepackage{amsfonts}
\usepackage{ifthen,float}
\usepackage{color}
\usepackage{rawfonts}
\usepackage{slashbox}
\usepackage{tipa}
\usepackage{mathrsfs}
\usepackage{graphicx,color}
\usepackage{latexsym}
\usepackage{amsmath,epsfig}
\usepackage{epic,eepic}
\newtheorem{proposition}{Proposition}

\newtheorem{theorem}{Theorem}
\newtheorem{lemma}{Lemma}
\newtheorem{corollary}{Corollary}

\begin{document}

\def\QEDclosed{\mbox{\rule[0pt]{1.3ex}{1.3ex}}}
\def\proof{\noindent\hspace{2em}{\itshape Proof: }}
\def\endproof{\hspace*{\fill}~\QEDclosed\par\endtrivlist\unskip}

\title{A new sufficient condition for sum-rate tightness in quadratic Gaussian multiterminal source coding}
\author{Yang Yang, Yifu Zhang, and Zixiang Xiong
\thanks{This work was supported by the NSF grant 0729149
and the Qatar National Research Fund.
The authors are with the department of
electrical and computer engineering, Texas A\&M University, College
Station, TX 77843. Emails: yangyang@tamu.edu, zyf674@tamu.edu, and
zx@ece.tamu.edu. Part of this work was presented at the Information Theory
and Applications Workshop, San Diego, CA, February 2010.}}

\maketitle 
\begin{abstract}
This work considers the quadratic Gaussian multiterminal (MT)
source coding problem and provides a new sufficient
condition for the Berger-Tung sum-rate bound to be tight. The
converse proof utilizes a set of virtual remote sources given which the MT sources are block independent with a maximum block size of two.
The given MT source coding problem is then related
to a set of two-terminal problems with matrix-distortion
constraints, for which a new lower bound on the sum-rate is given.
Finally, a convex optimization problem is formulated and a sufficient condition derived for the optimal BT scheme to satisfy the
subgradient based Karush-Kuhn-Tucker condition.
The set of sum-rate tightness problems defined by our
new sufficient condition subsumes all previously known tight
cases, and opens new direction for a more general partial solution.
\end{abstract}


\newpage

\section{Introduction}

Multiterminal (MT) source coding, which was introduced by Berger \cite{Berger77} and Tung \cite{Tung77} in 1977, defines the problem of separate compression and joint decompression of multiple correlated sources subject to distortion constraints. Finding the achievable rate region for the general MT problem is very hard, hence research has been focusing on the quadratic Gaussian case when the sources are jointly Gaussian and the distortion measure is the mean-squared error. The sum-rate part of the achievable rate region of the quadratic Gaussian MT problem is of particular interest and has been characterized for several special instances.

By connecting the quadratic Gaussian MT source coding problem to the quadratic Gaussian CEO problem \cite{VisBergerQGCEO,Oohama05}, Wagner {\em et al.} \cite{Wagner05} showed sum-rate tightness of the Berger-Tung (BT) rate region for the two-terminal and positive-symmetric cases. Wang {\em et al.} \cite{WangISIT09} then provided an alternative proof based on an estimation-theoretic result, which also leads to a sufficient condition for BT sum-rate tightness. Yang and Xiong \cite{Allerton:09} started with a generalized quadratic Gaussian CEO problem and proved sum-rate tightness in the bi-eigen equal-variance with equal distortion (BEEV-ED) case. Although the BEEV-ED case satisfies the sufficient condition given in \cite{WangISIT09}, the proof technique for the converse theorem is different and examples more explicit.

Wang {\em et al.}'s sufficient condition \cite{WangISIT09} is so far the most inclusive condition for BT sum-rate tightness, and its converse proof consists of the following steps. First, a set of $L$ virtual sources, referred to as the remote sources, were constructed such that the given $L$ MT sources can be viewed as independently Gaussian corrupted versions of the remote sources. Then, they used an estimation-theoretic result in conjunction with the semidefinite partial ordering of the distortion matrices to give a lower bound on the MT sum-rate. Finally, an optimization problem was formulated to find the best lower bound over possible (conditional and unconditional) distortion matrices, with the Karush-Kuhn-Tucker (KKT) condition given and simplified to prove their main result. An important assumption that enables their proof is that the observation noises between the virtual remote sources and the MT sources are independent Gaussian with a {\em diagonal} covariance matrix. Since the rate-distortion function for independent Gaussian random sources is completely known, this assumption dramatically simplifies the lower bound and hence the optimization problem.

In this paper, we provide a new and more inclusive sufficient condition than Wang {\em et al.}'s \cite{WangISIT09} for BT sum-rate tightness\footnote{The conference version of this work appeared in \cite{YangITA10}.}. The main novelty is to consider a larger set of remote sources, such that the observation noises between the MT and remote sources have a {\em block-diagonal} covariance matrix, instead of a diagonal matrix as assumed in \cite{WangISIT09}. By restricting the noise covariance matrix to have $K$ $2\times 2$ diagonal blocks and $(L-2K)$ $1\times 1$ diagonal blocks, we build a connection between the $L$-terminal problem and $K$ two-terminal problems with {\em matrix-distortion} constraint.

Unfortunately, although the original quadratic Gaussian two-terminal source coding problem with {\em individual} distortion constraint has been completely solved \cite{Oohama97,Wagner05}, the exact minimum sum-rate for its variant with a matrix-distortion constraint is still unknown in general. A composite lower bound is already provided by Wagner {\em et al.} \cite{Wagner05}. We partially improve the composite lower bound in this paper, using a technique inspired by Wang {\em et al.}'s work \cite{WangISIT09}. It is shown that this improvement can be infinitely large in some extreme cases. However, our new lower bound does not always match the BT upper bound, leaving a bounded gap between them.

Our new lower bound for the matrix-distortion constrained two-terminal problem is then utilized to give a new sum-rate lower bound on the $L$-terminal problem. After forming an optimization problem to search for the best $L$-terminal lower bound, we characterize the generalized KKT condition based on the subgradient \cite{subgradientbook} of the objective function, which is convex, continuous, but non-differentiable. Finally, our new sufficient condition is obtained by simplifying the subgradient-based KKT condition. Examples with tight sum-rate bound are also given.

The set of sum-rate tightness problems defined by our new sufficient condition subsumes all previously known tight cases, including Wagner {\em et al.}'s positive-symmetric case, Wang {\em et al.}'s sufficient condition, and Yang and Xiong's BEEV-ED case, thanks to the following two novelties. First, replacing Wang {\em et al.}'s independence assumption on the observation noises with a block-independent one leads to a larger repertoire of remote sources that serve as the basic tools for deriving the sum-rate lower bound. Second, the partially improved composite bound for the matrix-distortion constrained two-terminal problem gives a wider range of subgradients, hence a more relaxed subgradient-based KKT condition.

It is worth noting that, our new condition even includes degraded cases where the target distortions are not simultaneously achieved in the optimal BT scheme. This is the first time a degraded case is proved to have a tight BT sum-rate bound.

In addition, the technique introduced in this paper might be further generalized to allow $3\times 3$ (or even larger) block size in the observation noise covariance matrix to yield even more new tight cases, if one can explicitly give a lower-bound on the corresponding matrix-distortion constrained three-terminal problem.

The rest of this paper is organized as follows. Section II gives the formal definition of the quadratic Gaussian MT source coding problem and reviews existing results on sum-rate tightness. Section III studies the two-terminal source coding problem with matrix-distortion constraint, and provides an improved lower bound on the sum-rate. Section IV states our main results on a new sufficient condition for sum-rate tightness, and presents a degraded example belonging to the block-degraded case that satisfies our new condition. Section V gives a simplified sufficient condition for the sum-rate tightness in the non-degraded cases, followed by two additional examples satisfying the simplified condition. Section VI concludes the paper.

\section{The quadratic Gaussian MT source coding problem and existing results on sum-rate tightness}

\subsection{The quadratic Gaussian MT source coding problem}

For any integer $L$, denote $\mathcal{L}=\{1,2,...,L\}$. Let $Y_\mathcal{L}=(Y_1,Y_2,...,Y_L)^T$ be a length-$L$ vector
Gaussian source with mean $\boldsymbol{0}$ and covariance matrix $\Sigma_{Y_\mathcal{L}}$.
Also denote $Y_{\mathcal{S}_k}$ as the length-$|\mathcal{S}_k|$ subvector of $Y_\mathcal{L}$ indexed by $\mathcal{S}_k$. For an integer $n$, let $\boldsymbol{Y}_{\mathcal{L}}=(Y_{\mathcal{L},1},Y_{\mathcal{L},2},...,Y_{\mathcal{L},n})$ be an $L\times n$ matrix with $Y_{\mathcal{L},i}$, $i=1,2,...,n$ being $n$ independent drawings of $Y_{\mathcal{L}}$. Also denote $\boldsymbol{Y}_{\mathcal{S}_k}$ as the $|\mathcal{S}_k|\times n$ submatrix of $\boldsymbol{Y}_\mathcal{L}$ with column indices $\mathcal{S}_k$. For any $L\times n$ random matrix $\boldsymbol{Y}_\mathcal{L}$ and any random object $\omega$, define the conditional covariance matrix of $\boldsymbol{Y}_\mathcal{L}$ given $\omega$ as
\begin{eqnarray}
\mathrm{cov}(\boldsymbol{Y}_\mathcal{L}|\omega)&\stackrel{\Delta}{=}&\frac{1}{n}E\Big[\big(\boldsymbol{Y}_{\mathcal{L}}-E(\boldsymbol{Y}_{\mathcal{L}}|\omega)\big) \big(\boldsymbol{Y}_{\mathcal{L}}-E(\boldsymbol{Y}_{\mathcal{L}}|\omega)\big)^T\Big].
\end{eqnarray}

Consider the task of separately compressing a length-$n$ block of sources $\boldsymbol{Y}_\mathcal{L}$ at $L$ encoders and jointly reconstructing $\boldsymbol{Y}_\mathcal{L}$ as $\hat{\boldsymbol{Y}}_\mathcal{L}$ at a central decoder subject to individual distortion constraints $D_\mathcal{L}=\{D_1,D_2,...,D_L\}$. This problem is known as the {\em quadratic Gaussian MT source coding problem}, whose block diagram is depicted in Fig \ref{fig:diagram}.

\begin{figure}[t]
\centerline{\epsfxsize=4.2in
\epsfbox{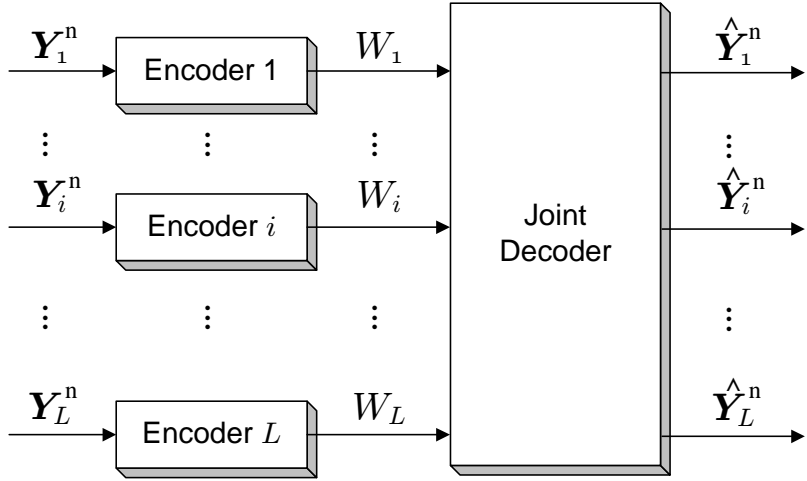}}
\caption{The quadratic Gaussian MT source coding problem.}
\label{fig:diagram}
\end{figure}

Let \begin{eqnarray}
\phi_j^{(n)}:\mathbb{R}^{n}\mapsto\{1,2,...,2^{R_j^{(n)}}-1\},~~j\in\mathcal{L}
\end{eqnarray}
be the $j$-th encoder function and
\begin{eqnarray}
\psi_j^{(n)}:\{1,2,...,2^{R_1^{(n)}}-1\}\times \{1,2,...,2^{R_2^{(n)}}-1\}
\times ...\times \{1,2,...,2^{R_L^{(n)}}-1\}\mapsto\mathbb{R}^n
\end{eqnarray}
be the reconstruction function for $\boldsymbol{Y}_j$.
Denote $W_j$ as the transmitted symbol at the $j$-th encoder,
and $R_{sum}(\phi^{(n)}_\mathcal{L},\psi^{(n)}_\mathcal{L})
=\sum_{j\in\mathcal{L}} R_j^{(n)}$ as the sum-rate of the MT coding scheme
$(\phi^{(n)}_\mathcal{L},\psi^{(n)}_\mathcal{L})$.
We say a rate tuple $(R_1,...,R_L)^T$ is $(\Sigma_{Y_\mathcal{L}},D_\mathcal{L})$-achievable if there exists
a sequence of schemes
$\{(\phi^{(n)}_\mathcal{L},\psi^{(n)}_\mathcal{L}):n\in\mathbb{N}^+\}$ such that
\begin{eqnarray}
\limsup_{n\rightarrow\infty} R_j^{(n)}&\le&R_j,
\mathrm{~~for~any~}j\in\mathcal{L},\\
\limsup_{n\rightarrow\infty}
\frac{1}{n}E\Big[({Y}_{j,i}-\hat{{Y}}_{j,i})^2\Big]&\le&D_j,
\mathrm{~~for~any~}j\in\mathcal{L}.
\end{eqnarray}
Define the
$(\Sigma_{Y_\mathcal{L}},D_\mathcal{L})$-achievable rate region
${\mathcal{R}}_{\Sigma_{Y_\mathcal{L}}}({D}_\mathcal{L})$
as the convex closure of all
$(\Sigma_{Y_\mathcal{L}},{D}_\mathcal{L})$-achievable
rate tuples, i.e.,
\begin{eqnarray}
{\mathcal{R}}_{\Sigma_{Y_\mathcal{L}}}^{}({D}_\mathcal{L})
&=& \mathrm{cl}\{(R_1,R_2, \ldots, R_L)^T: (R_1,R_2, \ldots, R_L)^T
\mathrm{~is~}
({\Sigma}_{{Y}},D_\mathcal{L})\mathrm{~achievable}\}.
\end{eqnarray}
The {\it minimum sum-rate} with respect to
$(\Sigma_{Y_\mathcal{L}},{D}_\mathcal{L})$ is then
defined as
\begin{eqnarray}
R_{\boldsymbol\Sigma_{Y_\mathcal{L}}}^{}(D_\mathcal{L})
&=& \inf\{~\sum_{i=1}^L R_i: (R_1,R_2, \ldots, R_L)^T \in
{\mathcal{R}}_{\boldsymbol\Sigma_{Y_\mathcal{L}}}^{}(D_\mathcal{L})\}.
\end{eqnarray}

Berger and Tung \cite{Berger77,Tung77} provide an {\it inner rate region} inside which all rate tuples are $({\Sigma}_{{Y}_\mathcal{L}},{D}_\mathcal{L})$-achievable.
In this paper, we restrict ourselves to a subset of the Berger-Tung inner rate region inside which all points can be achieved by parallel Gaussian test channels. This subset is referred to as the {\em Berger-Tung (BT)} inner rate region in the sequel. Let ${U}_\mathcal{L} = (U_1,U_2,\ldots,U_L)^T$ be a length-$L$
auxiliary random vector such that
\begin{itemize}
\item $U_i = Y_i + Q_i, i=1,2,\ldots, L$, where
$Q_i\sim\mathcal{N}(0,\sigma_{Q_i}^2)$, and all $Q_i$'s are
independent of each other and of all $Y_i$'s,
\item ${U}_\mathcal{L}$
satisfies $E\Big\{\big(Y_i-E(Y_i|{U}_\mathcal{L})\big)^2\Big\}\le
D_i$ for all $i=1,2,\ldots, L$,
\end{itemize}
and define
$\mathcal{U}({\Sigma}_{{Y}_\mathcal{L}},{D}_\mathcal{L})$
as the set of all auxiliary random vectors ${U}_\mathcal{L}$ that
satisfy the above conditions. Then the following
lemma gives the BT inner rate region, the proof can be found in
\cite{Berger77,Tung77}.

\begin{lemma}
\label{lemmaBT} Define \begin{eqnarray}
{\mathcal{R}}^{BT}_{{\Sigma}_{{Y}_\mathcal{L}}}({D}_\mathcal{L})
&=&\bigcup_{{U}_\mathcal{L}\hspace{0.01in}\in\hspace{0.01in}\mathcal{U}({\Sigma}_{{Y}_\mathcal{L}},{D}_\mathcal{L})}
\Big\{(R_1,R_2,\ldots,R_L)^T: \sum_{i\in \mathcal{A}} R_i \ge
I({Y}_{\mathcal{A}};{U}_{\mathcal{A}}|{U}_{\mathcal{L}-
\mathcal{A}})\Big\},
\end{eqnarray}
then
\begin{eqnarray}
{\mathcal{R}}^{BT}_{{\Sigma}_{{Y}_\mathcal{L}}}({D}_\mathcal{L})
&\subseteq&
{\mathcal{R}}^{}_{{\Sigma}_{{Y}_\mathcal{L}}}({D}_\mathcal{L}).\label{GBTrateregion}
\end{eqnarray}
In particular, the {\it BT minimum sum-rate}
\begin{eqnarray}
R_{sum}^{BT}(\Sigma_{{Y}_\mathcal{L}},{D}_\mathcal{L})
&=& \inf\{~\sum_{i=1}^L R_i: (R_1,R_2, \ldots, R_L)^T \in
{\mathcal{R}}_{{\Sigma}_{{Y}_\mathcal{L}}}^{BT}({D}_\mathcal{L})\}\nonumber\\
&=&
\inf_{ \Sigma_{Q_\mathcal{L}}\in\bblambda_L:\Big[(\Sigma_{Y_\mathcal{L}}^{-1}
+\Sigma_{Q_\mathcal{L}}^{-1})^{-1}\Big]_{j,j}\le D_j,
~\forall j\in\mathcal{L}}\hspace{-0.2in}\frac{1}{2}
\log_2\Big[\frac{|\Sigma_{Y_\mathcal{L}}|}{|(\Sigma_{Y_\mathcal{L}}^{-1}
+\Sigma_{Q_\mathcal{L}}^{-1})^{-1}|}\Big]
\end{eqnarray}
satisfies
\begin{eqnarray}
R_{sum}(\Sigma_{Y_\mathcal{L}},D_\mathcal{L})
&\le& R_{sum}^{BT}(\Sigma_{Y_\mathcal{L}},D_\mathcal{L}),
\end{eqnarray}
where $\bblambda_L$ denotes the set of all $L\times L$ positive definite (p.d.) diagonal matrices.
\end{lemma}

For example, the BT rate region for the quadratic Gaussian two-terminal source coding problem with $\Sigma_{Y_\mathcal{L}}=\left[\begin{array}{cc}\sigma_{Y_1}^2&\rho\sigma_{Y_1}\sigma_{Y_2}\\\rho\sigma_{Y_1}\sigma_{Y_2}&\sigma_{Y_2}^2\end{array}\right]$ is given by
\begin{eqnarray}
{\mathcal{R}}^{BT}_{{\Sigma}_{{Y}_\mathcal{L}}}({D}_\mathcal{L})=
\hat{\mathcal{R}}^{BT}_1(D_1,D_2)\cap
\hat{\mathcal{R}}^{BT}_2(D_1,D_2) \cap
\hat{\mathcal{R}}^{BT}_{12}(D_1,D_2),
\end{eqnarray}
where
\begin{eqnarray}
\hat{\mathcal{R}}^{BT}_i(D_1,D_2) &=&
\{(R_1,R_2):R_i\ge\frac{1}{2}\log^+[(1-\rho^2+\rho^2 2^{-2R_j})
\frac{\sigma_{Y_i}^2}{D_i}]\}, i,j=1,2, i\ne j, \label{DMT-single}\\
\hat{\mathcal{R}}^{BT}_{12}(D_1,D_2) &=& \{(R_1,R_2):
R_1+R_2\ge\frac{1}{2}\log^+[(1-\rho^2)
\frac{\beta_{max}\sigma_{Y_1}^2 \sigma_{Y_2}^2}{2D_1 D_2}]\},
\label{DMTR}
\end{eqnarray}
with $\beta_{max}=1+\sqrt{1+\frac{4\rho^2 D_1 D_2}{(1-\rho^2)^2
\sigma_{Y_1}^2 \sigma_{Y_2}^2}}$, and $\log^+ x=\max\{\log x, 0\}$.
The BT rate region with $\sigma_{Y_1}^2=\sigma_{Y_2}^2=1$, $\rho=0.9$, and $D_\mathcal{L}=(0.1,0.1)^T$ is shown in Fig. \ref{fig:region}, where $\partial\hat{\mathcal{R}}^{BT}_i(D_1,D_2)$ and $\partial\hat{\mathcal{R}}^{BT}_{12}(D_1,D_2)$ are the boundaries of $\hat{\mathcal{R}}^{BT}_i(D_1,D_2)$ and $\hat{\mathcal{R}}^{BT}_{12}(D_1,D_2)$, respectively.

\begin{figure}[tbh]
\centerline{\epsfxsize=3.8in
\epsfbox{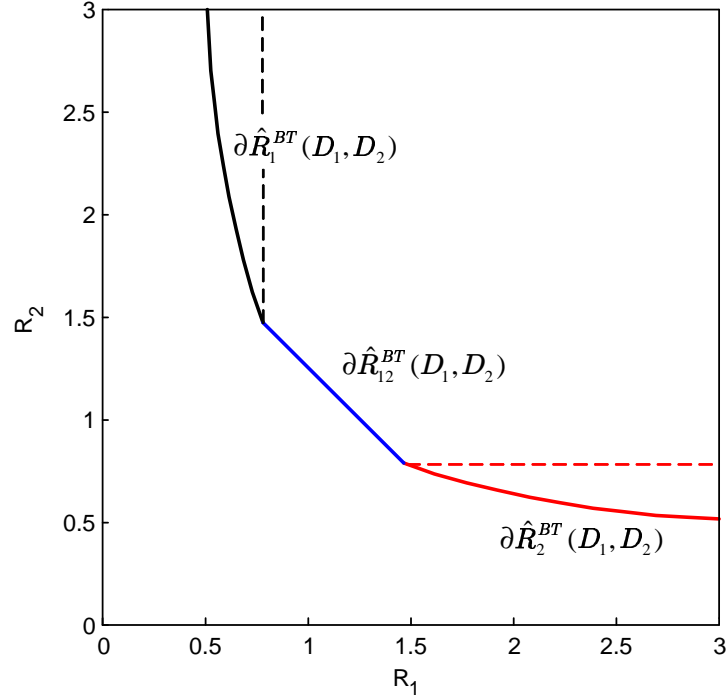}}
\caption{An example of the BT rate region for the quadratic Gaussian two-terminal source coding problem.}
\label{fig:region}
\end{figure}

\subsection{Existing results on sum-rate tightness}

Wagner {\em et al.} \cite{Wagner05} proved that for the two-terminal
case (with $L=2$), the BT minimum sum-rate is equal to the MT minimum sum-rate,
i.e.,
\begin{eqnarray}
R_{sum}(\Sigma_{Y_\mathcal{L}},D_\mathcal{L})
&=& R_{sum}^{BT}(\Sigma_{Y_\mathcal{L}},D_\mathcal{L})\label{Wagnertightcases}
\end{eqnarray}
for any $2\times 2$ positive semidefinite (p.s.d.) symmetric matrix
$\Sigma_{Y_\mathcal{L}}$ and length-$2$ positive vector $D_\mathcal{L}$.
They also showed tightness of the BT sum-rate bound for the positive
symmetric case, i.e., (\ref{Wagnertightcases}) holds for any $L\times L$
{\em positive-symmetric} matrices of the form
\begin{eqnarray}
\Sigma_{Y_\mathcal{L}}=\mathcal{S}_L(a,b)\stackrel{\Delta}{=}
\left[\begin{array}{ccccc}
a&b&b&...&b\\
b&a&b&...&b\\
...&...&...&...&...\\
b&b&b&...&a\\\end{array}\right],\label{symmetricmat}
\end{eqnarray}
for some $a>b>0$ and any $D_\mathcal{L}=(D,D,...,D)^T$
for some $D>0$.

The most general cases of quadratic Gaussian MT source coding problem with tight sum-rate are provided by Wang {\em et al.} \cite{WangISIT09}. Their proof contains four major steps.
\begin{itemize}
\item First, the $L$ MT sources $Y_\mathcal{L}$ are connected to $L$ remote sources $X_\mathcal{L}$ such that
\begin{eqnarray}
Y_\mathcal{L}&=&X_\mathcal{L}+N_\mathcal{L}
\end{eqnarray}
with $N_\mathcal{L}$ being a zero-mean Gaussian vector independent of $X_\mathcal{L}$ with a {\em diagonal} covariance matrix \begin{eqnarray}
\Sigma_{N_\mathcal{L}}&=&\left[\begin{array}{cccc}\sigma_{N_1}^2&0&...&0\\0&\sigma_{N_2}^2&...&0\\...&...&...&...\\0&0&...&\sigma_{N_L}^2\end{array}\right].\label{SigmaND}
\end{eqnarray}
Then they use the Markov chain $\boldsymbol{X}_\mathcal{L}\rightarrow \boldsymbol{Y}_\mathcal{L}\rightarrow W_\mathcal{L}$ to obtain an {\em estimation-theoretic result} that $\mathrm{cov}({\boldsymbol{Y}_\mathcal{L}|\boldsymbol{X}_\mathcal{L},W_\mathcal{L}})$ must also be diagonal.

\item Exploit the {\em semidefinite partial order} of the distortion matrices, which is due to the fact that a linear {\em minimum mean squared error} (MMSE) estimator cannot outperform its optimal MMSE counterpart, to show that
 \begin{eqnarray}
 \mathrm{cov}({\boldsymbol{Y}_\mathcal{L}|\boldsymbol{X}_\mathcal{L},W_\mathcal{L}})~\preceq~ \Big(\big(\mathrm{cov}({\boldsymbol{Y}_\mathcal{L}|W_\mathcal{L}})\big)^{-1}+\Sigma_{N_\mathcal{L}}^{-1}-\Sigma_{Y_\mathcal{L}}^{-1}\Big)^{-1}.\nonumber 
 \end{eqnarray}

\item A lower bound on the MT minimum sum-rate $R_{\boldsymbol\Sigma_{Y_\mathcal{L}}}^{}(D_\mathcal{L})$ is derived by exploiting the diagonal structure of $\mathrm{cov}({\boldsymbol{Y}_\mathcal{L}|\boldsymbol{X}_\mathcal{L},W_\mathcal{L}})$.

\item Form a convex optimization problem that minimizes the above lower bound over $\boldsymbol{D}\stackrel{\Delta}{=}\mathrm{cov}({\boldsymbol{Y}_\mathcal{L}|W_\mathcal{L}})$ and $\gamma_\mathcal{L}\stackrel{\Delta}{=}\mathrm{diag}\Big(\mathrm{cov}({\boldsymbol{Y}_\mathcal{L}|\boldsymbol{X}_\mathcal{L},W_\mathcal{L}})\Big)$, and establish a sufficient condition for the $\boldsymbol{D}$ and $\gamma_\mathcal{L}$ that correspond to the optimal BT scheme to satisfy the the KKT condition of the optimization problem.
\end{itemize}

Specifically, let $\mathscr{P}^\succeq_L$ be the set of $L\times L$ p.s.d. matrices and $\mathbbm{d}$ be the set of diagonal matrices.
Define $\mathscr{D}(D_\mathcal{L},\Sigma_{Y_\mathcal{L}})$ as the set of all BT-achievable distortion matrices that satisfy the distortion constraints, and $\mathscr{N}(\Sigma_{Y_\mathcal{L}})$ as the set of all possible diagonal covariance matrices $\Sigma_{N_\mathcal{L}}$, i.e.,
\begin{eqnarray}
\mathscr{D}(D_\mathcal{L},\Sigma_{Y_\mathcal{L}})&\stackrel{\Delta}
{=}&\Big\{\boldsymbol{D}\in\mathbb{R}^{L\times L}: [\boldsymbol{D}]_{j,j}
=D_j,\forall j\in\mathcal{L}, \mathrm{~and~} \boldsymbol{D}^{-1}
-\Sigma_{Y_\mathcal{L}}^{-1}\in\mathscr{P}^\succeq\cap\mathbbm{d}\Big\},\\
\mathscr{N}(\Sigma_{Y_\mathcal{L}})&\stackrel{\Delta}
{=}&\left\{\Sigma\in\mathscr{P}^\succeq\cap\mathbbm{d}: \Sigma\succeq\Sigma_{Y_\mathcal{L}}\right\}.\label{DNdef}
\end{eqnarray}
Wang {\em et al.}'s result \cite{WangISIT09} is summarized in the following
theorem.

\begin{theorem}[\cite{WangISIT09}]\label{thmWang}
If for some $\boldsymbol{D}\in\mathscr{D}(D_\mathcal{L},\Sigma_{Y_\mathcal{L}})$
and $\Sigma_{N_\mathcal{L}}\in\mathscr{N}(\Sigma_{Y_\mathcal{L}})$,
there exists a diagonal matrix $\Pi=\mathrm{diag}(\pi_1,...,\pi_L)$
such that
\begin{eqnarray}
\boldsymbol{D}\Big(\Pi-\boldsymbol{D}^{-1}+\boldsymbol{D}^{-1}(\boldsymbol{D}^{-1}+\Sigma_{N_\mathcal{L}}^{-1}-\Sigma_{Y_\mathcal{L}}^{-1})^{-1}\boldsymbol{D}^{-1}\Big)\boldsymbol{D}
\end{eqnarray} is a p.s.d. matrix with the same diagonal
elements as those of $(\boldsymbol{D}^{-1}+\Sigma_{N_\mathcal{L}}^{-1}-\Sigma_{Y_\mathcal{L}}^{-1})^{-1}$, then the BT sum-rate bound is tight, i.e.,
\begin{eqnarray}
R_{sum}(\Sigma_{Y_\mathcal{L}},D_\mathcal{L})
&=&R_{sum}^{BT}(\Sigma_{Y_\mathcal{L}},D_\mathcal{L}).
\end{eqnarray}
\end{theorem}

Using a different technique, sum-rate tightness for a special
{\em bi-eigen equal-variance with equal distortion} class of MT problems was proved by Yang
and Xiong \cite{Allerton:09}. That is, (\ref{Wagnertightcases}) holds
for any $\Sigma_{Y_\mathcal{L}}\in\mathcal{B}$ and
$D_\mathcal{L}=(D,D,...,D)^T$ for some $D>0$,
where $\mathcal{B}$ denotes the set of all $L\times L$ p.s.d. matrices with two distinct eigenvalues and
equal diagonal elements.

\section{The two-terminal source coding problem with a matrix-distortion constraint}

In order to go beyond Wang {\em et al.}'s sufficient condition \cite{WangISIT09}, which assumes independent observation noises as seen in (\ref{SigmaND}) and is derived using classical Gaussian rate-distortion function, in this paper we allow $2\times 2$ block-correlation among the observation noises. Consequently, the derivation of the new lower bound requires us to consider a variant of the two-terminal source coding problem where the two individual distortion constraints are replaced by a $2\times 2$ matrix-distortion constraint. Although the original quadratic Gaussian two-terminal source coding problem is completely solved \cite{Oohama97,Wagner05}, due to the different distortion constraints, the exact achievable rate region for the matrix-distortion constrained two-terminal problem is still unknown. In this section, we derive a lower bound on the sum-rate of the matrix-distortion constrained two-terminal problem, which serves as the key to our main results given in the next section.

Assume that length-$n$ blocks of Gaussian sources $\boldsymbol{Y}_1$ and $\boldsymbol{Y}_2$ are separated compressed at the two encoders, while the decoder tries to reconstruct $\boldsymbol{Y}_\mathcal{L}$ such that
\begin{eqnarray}
\limsup_{n\rightarrow\infty}\frac{1}{n}\sum_{i=1}^nE\Big[({Y}_{\mathcal{L},i}-\hat{{Y}}_{\mathcal{L},i})({Y}_{\mathcal{L},i}-\hat{{Y}}_{\mathcal{L},i})^T\Big] &\preceq&\boldsymbol{D}~=~ \left[\begin{array}{cc}D_{1}&\theta\sqrt{D_1D_2}\\\theta\sqrt{D_1D_2}&D_2\end{array}\right],
\end{eqnarray}
where $\boldsymbol{A}\preceq\boldsymbol{B}$ means $\boldsymbol{B}-\boldsymbol{A}$ is a p.s.d. matrix, and denote the minimum sum-rate of such a problem as $R_{sum}(\Sigma_{Y_\mathcal{L}},\boldsymbol{D})$. Compared to the original quadratic Gaussian two-terminal source coding problem with individual distortion constraints, we have
\begin{eqnarray}
R_{sum}(\Sigma_{Y_\mathcal{L}},(D_1,D_2)^T)&=&\inf_{\theta\in[-1,1]}R_{sum}\left(\Sigma_{Y_\mathcal{L}},\left[\begin{array}{cc}D_1&\theta\sqrt{D_1D_2}\\
\theta\sqrt{D_1D_2}&D_2\end{array}\right]\right).\label{dDrelation}
\end{eqnarray}

Although Wagner {\em et al.}'s paper \cite{Wagner05} focused on the original quadratic Gaussian two-terminal source coding problem, their converse proof has already explored the relationship in (\ref{dDrelation}) to some extent, and provided a composite lower bound on the sum-rate of the two-terminal source coding problem with matrix-distortion constraint, namely,
\begin{eqnarray}
R_{sum}(\Sigma_{Y_\mathcal{L}},\boldsymbol{D})&\ge&\max\Big\{R_{coop}(\Sigma_{Y_\mathcal{L}},\boldsymbol{D}),R_{\mu}(\Sigma_{Y_\mathcal{L}},\boldsymbol{D})\Big\},\label{Wagner2mat}
\end{eqnarray}
where
\begin{eqnarray}
R_{coop}(\Sigma_{Y_\mathcal{L}},\boldsymbol{D})&=&\frac{1}{2}\log\frac{|\Sigma_{Y_\mathcal{L}}|}{|\boldsymbol{D}|},\nonumber\\
R_{\mu}(\Sigma_{Y_\mathcal{L}},\boldsymbol{D})&=&R_{\Sigma_{Y_\mathcal{L}},\mu}({\tilde{\mu}}^T\boldsymbol{D}\tilde{\mu}),\nonumber
\end{eqnarray}
$\tilde{\mu}=(\sqrt{D_2},\sqrt{D_1})^T$, and $R_{\Sigma_{Y_\mathcal{L}},\mu}(d)$ denotes the minimum sum-rate of the $\mu$-sum problem with target distortion $d$.

We now give the exact form of a new lower bound that is inspired by Wang {\em et al.}'s work \cite{WangISIT09} and partially tighter than Wagner {\em et al.}'s bound in (\ref{Wagner2mat}). Note that there is no loss in assuming that the correlation coefficient $\rho$ between $Y_1$ and $Y_2$ is non-negative.

\begin{lemma}\label{lemma2term}
For any pair of $2\times 2$ matrices
\begin{eqnarray}
\Sigma_{Y_\mathcal{L}}&=&\left[\begin{array}{cc}\sigma_{Y_1}^2&\rho\sigma_{Y_1}\sigma_{Y_2}\\\rho\sigma_{Y_1}\sigma_{Y_2}&\sigma_{Y_2}^2\end{array}\right],\label{SigmaYL2x2}\\
\boldsymbol{D}&=&\left[\begin{array}{cc}D_{1}&\theta\sqrt{D_1D_2}\\\theta\sqrt{D_1D_2}&D_2\end{array}\right]\label{D2x2}
\end{eqnarray}
such that
\begin{eqnarray}
\rho\ge 0, \mathrm{~and~}\boldsymbol{D}\preceq\Sigma_{Y_\mathcal{L}},\label{SigmaD2x2cond}
\end{eqnarray}
it holds that
\begin{eqnarray}
R_{sum}(\Sigma_{Y_\mathcal{L}},\boldsymbol{D})
&\ge&\underline{R}_{sum}(\Sigma_{Y_\mathcal{L}},\boldsymbol{D})\nonumber\\
&\stackrel{\Delta}{=}&\max\Big\{R_{lb}(\Sigma_{Y_\mathcal{L}},\boldsymbol{D}), R_{\mu}(\Sigma_{Y_\mathcal{L}},\boldsymbol{D})\Big\}\nonumber\\
&=&\left\{\begin{array}{cc}R_{\mu}(\Sigma_{Y_\mathcal{L}},\boldsymbol{D})&\theta\le\tilde{\theta}\\ R_{lb}(\Sigma_{Y_\mathcal{L}},\boldsymbol{D})&\theta>\tilde{\theta}\end{array}\right.,\label{Lemma1statement}
\end{eqnarray}
where
\begin{eqnarray}
R_{\mu}(\Sigma_{Y_\mathcal{L}},\boldsymbol{D})&=&\frac{1}{2}\log\frac{v_1v_2(v_1v_2(1-\rho^2)+2\rho(1+\theta))}{(1+\theta)^2}\nonumber\\
R_{lb}(\Sigma_{Y_\mathcal{L}},\boldsymbol{D})&=&\frac{1}{2}\log\frac{v_{1}^3v_{2}^3(1-\rho^2)^2}{(1-\theta)^2(v_1v_2(1-\rho^2)+2\rho(1+\theta))},
\end{eqnarray}
with $v_1=\frac{\sigma_{Y_1}}{\sqrt{D_1}}$, $v_2=\frac{\sigma_{Y_2}}{\sqrt{D_2}}$, and
\begin{eqnarray}
\tilde{\theta}&\stackrel{\Delta}{=}&\frac{\sqrt{v_1^2v_2^2(1-\rho^2)^2+4\rho^2}-v_1v_2(1-\rho^2)}{2\rho}.\label{thetastardef}
\end{eqnarray}
Particularly, if $\theta\le\tilde{\theta}$, the lower bound is tight, i.e., $R_{sum}(\Sigma_{Y_\mathcal{L}},\boldsymbol{D})=\underline{R}_{sum}(\Sigma_{Y_\mathcal{L}},\boldsymbol{D})$.
\end{lemma}

\begin{proof}
See Appendix A.
\end{proof}

Note that unlike the original two-terminal problem, the new lower bound $\underline{R}_{sum}(\Sigma_{Y_\mathcal{L}},\boldsymbol{D})$ does not always meet the BT upper bound, which is given by
\begin{eqnarray}
R_{sum}^{BT}(\Sigma_{Y_\mathcal{L}},\boldsymbol{D})&=&\max\Big\{R_{lb}(\Sigma_{Y_\mathcal{L}},\boldsymbol{D}), R_{\mu}(\Sigma_{Y_\mathcal{L}},\boldsymbol{D})\Big\}\nonumber\\
&=&\left\{\begin{array}{cc}R_{\mu}(\Sigma_{Y_\mathcal{L}},\boldsymbol{D})&\theta\le\tilde{\theta}\\ R_{ub}(\Sigma_{Y_\mathcal{L}},\boldsymbol{D})&\theta>\tilde{\theta}\end{array}\right.
\end{eqnarray}
with
\begin{eqnarray}
R_{ub}(\Sigma_{Y_\mathcal{L}},\boldsymbol{D})&=&\frac{1}{2}\log\frac{v_1v_2(v_1v_2(1-\rho^2)-2\rho(1-\theta))}{(1-\theta)^2}.
\end{eqnarray}
Obviously, if $\theta>\tilde{\theta}$, the two bounds do not coincide, and we can easily compute the gap between them as
\begin{eqnarray}
R^\Delta_{sum}(\Sigma_{Y_\mathcal{L}},\boldsymbol{D})&\stackrel{\Delta}{=}& \underline{R}_{sum}(\Sigma_{Y_\mathcal{L}},\boldsymbol{D})-R_{sum}^{BT}(\Sigma_{Y_\mathcal{L}},\boldsymbol{D})\nonumber\\&=& R_{ub}(\Sigma_{Y_\mathcal{L}},\boldsymbol{D})-R_{lb}(\Sigma_{Y_\mathcal{L}},\boldsymbol{D})\nonumber\\
&=&\frac{1}{2}\log\frac{(v_1v_2(1-\rho^2)-2\rho(1-\theta))(v_1v_2(1-\rho^2)+2\rho(1+\theta))}{v_1^2v_2^2(1-\rho^2)^2}.
\end{eqnarray}
To evaluate the maximum value of $R^\Delta_{sum}(\Sigma_{Y_\mathcal{L}},\boldsymbol{D})$, we compute the feasible range of $\theta$, which is constrained by the assumption $\boldsymbol{D}\preceq\Sigma_{Y_\mathcal{L}}$, and given by $\theta\in(\underline{\theta},\overline{\theta})$ with
\begin{eqnarray}
\underline{\theta}&=&\max\left\{-1,-\sqrt{(v_1^2-1)(v_2^2-1)}-\rho v_1v_2\right\},\nonumber\\
\overline{\theta}&=&\min\left\{1,\sqrt{(v_1^2-1)(v_2^2-1)}+\rho v_1v_2\right\}.\label{overthetadef}
\end{eqnarray}
Now due to the assumption that $\rho\ge 0$, $R^\Delta_{sum}(\Sigma_{Y_\mathcal{L}},\boldsymbol{D})$ is monotone increasing in $\theta$ in the range $(\tilde{\theta},\overline{\theta})$. Hence
\begin{eqnarray}
\sup_{\theta\in(\tilde{\theta},\overline{\theta})}R^\Delta_{sum}(\Sigma_{Y_\mathcal{L}},\boldsymbol{D})&=&\lim_{\theta\rightarrow \overline{\theta}}R^\Delta_{sum}(\Sigma_{Y_\mathcal{L}},\boldsymbol{D})\nonumber\\ &\le&\lim_{\theta\rightarrow 1}R^\Delta_{sum}(\Sigma_{Y_\mathcal{L}},\boldsymbol{D})\nonumber\\&=& \frac{1}{2}\log\left(1+\frac{4\rho}{v_1v_2(1-\rho^2)}\right).
\end{eqnarray}
We thus conclude that although the lower bound $\underline{R}_{sum}(\Sigma_{Y_\mathcal{L}},\boldsymbol{D})$ is not always tight, the gap to the upper bound $R_{sum}^{BT}(\Sigma_{Y_\mathcal{L}},\boldsymbol{D})$ cannot exceed a certain threshold that depends only on $v_1$, $v_2$, and $\rho$.

On the other hand, if we calculate the improvement from Wagner {\em et al.}'s lower bound (\ref{Wagner2mat}) to our new one $\underline{R}_{sum}(\Sigma_{Y_\mathcal{L}},\boldsymbol{D})$ with $\theta\in(\tilde{\theta},\overline{\theta})$, we obtain
\begin{eqnarray}
&&\underline{R}_{sum}(\Sigma_{Y_\mathcal{L}},\boldsymbol{D})-\max\Big\{R_{coop}(\Sigma_{Y_\mathcal{L}},\boldsymbol{D}),R_{\mu}(\Sigma_{Y_\mathcal{L}},\boldsymbol{D})\Big\}\nonumber\\
&=&\frac{1}{2}\log\frac{(v_1v_2(1+\theta)(1-\rho^2)}{(1-\theta)(v_1v_2(1-\rho^2)+2\rho(1+\theta))},
\end{eqnarray}
which obviously goes to infinity as $\theta\rightarrow 1$, this means that the improvement can be infinitely large for any value of $v_1$, $v_2$, and $\rho$ such that $\overline{\theta}$ defined in (\ref{overthetadef}) equals to one.

A comparison among Wagner's lower bound \cite{Wagner05}, our partially improved lower bound, and the BT upper bound with $\sigma_{Y_1}^2=\sigma_{Y_2}^2=1$, $\rho=0.9$, $D_1=0.1$, $D_2=0.05$ is shown in Fig. \ref{fig:compare}. We can clearly observe that the gap from our new lower bound to the BT upper bound is much smaller than that to the lower bound in \cite{Wagner05}.

\begin{figure}[tbh]
\centerline{\epsfxsize=5.8in
\epsfbox{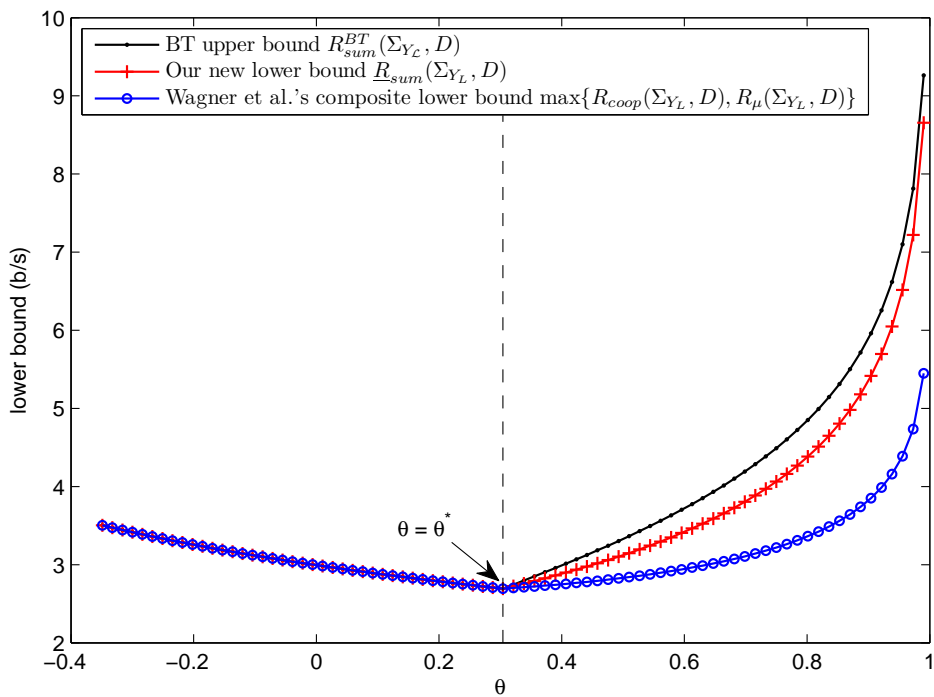}}
\caption{Comparison among Wagner's lower bound \cite{Wagner05}, our partially improved lower bound, and the BT upper bound.}
\label{fig:compare}
\end{figure}

\section{Main results}

\subsection{Definitions and preliminaries}

Before stating our main results, we need to give some definitions and review the subgradient-based KKT condition.

Let $\pi=\{\pi_1,...,\pi_L\}$ be a permutation of $\mathcal{L}$, and $\bbpi$ be the corresponding $L\times L$ permutation matrix such that
$\bbpi\mathcal{L}=\pi.$ We say an $L\times L$ matrix $\Sigma$ is $\pi^{(K)}$ block-diagonal if it is symmetric and can be written as
\begin{eqnarray}
\Sigma=\bbpi\cdot\left[\begin{array}{cccccccccccc}a_{1,1}&a_{1,2}&0&0&...&...&0&0&0&...&0&0\\a_{1,2}&a_{2,2}&0&0&...&...&0&0&0&...&0&0\\ 0&0&a_{3,3}&a_{3,4}&...&...&0&0&0&...&0&0\\
0&0&a_{3,4}&a_{4,4}&...&...&0&0&0&...&0&0\\...&...&...&...&...&...&...&...&...&...\\0&0&0&0&...&...&a_{2K-1,2K-1}&a_{2K-1,2K}&0&...&0&0 \\0&0&0&0&...&...&a_{2K-1,2K}&a_{2K,2K}&0&...&0&0\\0&0&0&0&...&...&0&0&a_{2K+1}&...&0&0
\\0&0&0&0&...&...&0&0&0&...&a_{L-1}&0\\0&0&0&0&...&...&0&0&0&...&0&a_L\end{array}\right]\bbpi^T,\label{SigmaNBL}
\end{eqnarray}
and denote $\Upsilon_K(\pi)$ as the set of all $\pi^{(K)}$ block-diagonal matrices. Equivalently, $\Sigma\in\Upsilon_K(\pi)$ if and only if $\Sigma=\Sigma^T$ and
\begin{eqnarray}
\Sigma_{\pi_i,\pi_j}=0\mathrm{~if~}\left\{\begin{array}{c} i,j\in\{1,2,...,2K\}\mathrm{~such~that~}\lceil \frac{i}{2}\rceil\neq\lceil \frac{j}{2}\rceil\\
i,j\in\{2K+1,2K+2,...,L\}\mathrm{~such~that~}i\neq j\\i\in\{2K+1,2K+2,...,L\}\mathrm{~and~} j\in\{1,2,...,2K\}\\i\in\{1,2,...,2K\}\mathrm{~and~} j\in\{2K+1,2K+2,...,L\} \end{array}\right..
\end{eqnarray}

Comparing (\ref{SigmaNBL}) and (\ref{SigmaND}), it is clear that all diagonal matrices are also $\pi^{(K)}$ block-diagonal, but the converse is not true for $K\ge 1$, i.e., \begin{eqnarray}
\mathbbm{d}\subsetneq\Upsilon_K(\pi)\mathrm{~for~}1\le K\le\lfloor\frac{L}{2}\rfloor\mathrm{~ and~ any~ permutation~} \pi.
\end{eqnarray}
Consequently, if we define
\begin{eqnarray}
\mathscr{N}_{\pi^{(K)}}(\Sigma_{Y_\mathcal{L}})&\stackrel{\Delta}{=}&\left\{\Sigma\in\mathscr{P}^\succeq\cap\Upsilon_K(\pi): \Sigma\succeq\Sigma_{Y_\mathcal{L}}\right\},
\end{eqnarray}
and compare with $\mathscr{N}(\Sigma_{Y_\mathcal{L}})$ defined in (\ref{DNdef}), it holds that
\begin{eqnarray}
\mathscr{N}(\Sigma_{Y_\mathcal{L}})&=&\mathscr{N}_{\mathbb{I}^{(0)}}(\Sigma_{Y_\mathcal{L}})~\subseteq~\mathscr{N}_{\pi^{(K)}}(\Sigma_{Y_\mathcal{L}})
\end{eqnarray}
for any $0\le K\le\lfloor \frac{L}{2}\rfloor$ and permutation $\pi$, where $\mathbb{I}$ denotes the identity permutation that maps $\mathcal{L}$ to itself.

For a set of $L$ Gaussian sources $Y_\mathcal{L}$ and a $\Sigma_{N_\mathcal{L}}\in\Upsilon_K(\pi)$ such that $\Sigma_{N_\mathcal{L}}\preceq\Sigma_{Y_\mathcal{L}}$, let $M=\mathrm{rank}(\Sigma_{Y_\mathcal{L}}-\Sigma_{N_\mathcal{L}})$ and the singular value decomposition of $\Sigma_{Y_\mathcal{L}}-\Sigma_{N_\mathcal{L}}$ be
\begin{eqnarray}
\Sigma_{Y_\mathcal{L}}-\Sigma_{N_\mathcal{L}}=\boldsymbol{T}^T\mathrm{diag}(\sigma_{X_1}^2,\sigma_{X_2}^2,...,\sigma_{X_M}^2,0,...,0)\boldsymbol{T}.\label{Xconst1}
\end{eqnarray}
Then define $\Sigma_{X_\mathcal{M}}=\mathrm{diag}(\sigma_{X_1}^2,\sigma_{X_2}^2,...,\sigma_{X_M}^2)$, $\boldsymbol{H}=\boldsymbol{T}_{\mathcal{M},\mathcal{L}}$, and let
\begin{eqnarray}
X_\mathcal{M}&\stackrel{\Delta}{=}&\boldsymbol{A} Y_\mathcal{L} + Z_\mathcal{L},\label{Xconst2}
\end{eqnarray}
with $Z_\mathcal{L}\sim\mathcal{N}(\boldsymbol{0},\boldsymbol{B})$ independent of $Y_\mathcal{L}$, where
\begin{eqnarray}
\boldsymbol{A}&=&\Sigma_{X_\mathcal{M}}\boldsymbol{H}\Sigma_{Y_\mathcal{L}}^{-1},\nonumber\\
\boldsymbol{B}&=&\Sigma_{X_\mathcal{M}}-\Sigma_{X_\mathcal{M}}\boldsymbol{H}\Sigma_{Y_\mathcal{L}}^{-1}\boldsymbol{H}^T\Sigma_{X_\mathcal{M}}.\label{ABdef}
\end{eqnarray}
It is trivial to verify that the $M$ Gaussian {\em remote sources} $X_\mathcal{M}\sim\mathcal{N}(\boldsymbol{0},\Sigma_{X_\mathcal{M}})$ satisfy
\begin{eqnarray}
Y_\mathcal{L}&=&\boldsymbol{H}^T X_\mathcal{M}+N_\mathcal{L},\label{correlationstruct}
\end{eqnarray}
with the $L$ {\em observation noises} $N_\mathcal{L}\sim\mathcal{N}(\boldsymbol{0},\Sigma_{N_\mathcal{L}})$ independent of $X_\mathcal{M}$.

Next, we briefly review the subgradient-based KKT conditions for non-differentiable convex optimization problems. The original KKT condition is a necessary condition for global optimality in a convex optimization problem with differentiable objective function and equality/inequality constraints. However, when dealing with non-differentiable convex optimization problems, subgradient-based KKT condition has to be used instead.
We call $\boldsymbol{g}$ a subgradient \cite{subgradientbook} of a non-differentiable scalar-valued vector function $f$ at point $\boldsymbol{x}$, if
\begin{eqnarray}
f(\boldsymbol{y})&\ge&f(\boldsymbol{x})+\boldsymbol{g}^T(\boldsymbol{y}-\boldsymbol{x})~\mathrm{for~all~}\boldsymbol{y}.
\end{eqnarray}
In particular, if $f=\max\{f_1,f_2\}$ with $f_1$ and $f_2$ being convex and differentiable such that $f_1(\boldsymbol{x}_0)=f_2(\boldsymbol{x}_0)$, then the subgradients of $f$ at $\boldsymbol{x}_0$ form a line segment between $\nabla f_1(\boldsymbol{x}_0)$ and $\nabla f_2(\boldsymbol{x}_0)$. The set of all subgradients of a function $f$ at some point $\boldsymbol{x}$ is called the subdifferential of $f$ at $\boldsymbol{x}$, and denoted as $\partial f(\boldsymbol{x})$. The subdifferential of $\underline{R}_{sum}(\Sigma_{Y_{\mathcal{L}}},\boldsymbol{\Gamma})$ is given in the following lemma, with a detailed proof given in Appendix B.

\begin{lemma}\label{lemmasubdiff}
Assume that $\Sigma_{Y_\mathcal{L}}$ and $\boldsymbol{D}$ take forms of (\ref{SigmaYL2x2}) and (\ref{D2x2}), respectively, such that $\boldsymbol{D}\preceq\Sigma_{Y_\mathcal{L}}$. Then the subdifferential of $\underline{R}_{sum}(\Sigma_{Y_{\mathcal{L}}},\boldsymbol{D})$ (as a function of $\boldsymbol{D}$) at
\begin{eqnarray}
\boldsymbol{D}=\tilde{\boldsymbol{D}}&\stackrel{\Delta}{=}&\left[\begin{array}{cc}{D}_1&{\tilde{\theta}}\sqrt{{D}_1{D}_2}\\{\tilde{\theta}}\sqrt{{D}_1{D}_2}&{D}_2\end{array}\right]
\end{eqnarray}
is a line segment
\begin{eqnarray}
{\partial} \underline{R}_{sum}(\Sigma_{Y_{\mathcal{L}}},{\boldsymbol{D}})\mid_{\boldsymbol{D}=\tilde{\boldsymbol{D}}}&{=}& \Big\{-\frac{1}{2}\tilde{\boldsymbol{D}}^{-1}\Psi~\tilde{\boldsymbol{D}}^{-1}: \Psi\in\bbdelta(\Sigma_{Y_\mathcal{L}},D_\mathcal{L})\Big\},\nonumber
\end{eqnarray}
\begin{eqnarray}
\bbdelta(\Sigma_{Y_\mathcal{L}},D_\mathcal{L})\stackrel{\Delta}{=} \left\{ \left[\begin{array}{cc}\hspace{-0.1in}{D}_1\hspace{-0.1in}&\hspace{-0.1in}\big(\alpha +(1-\alpha)(2{|\tilde{\theta}|}-1)\big)s\sqrt{{D}_1{D}_2}\hspace{-0.1in}\\ \big(\alpha +(1-\alpha)(2{|\tilde{\theta}|}-1)\big)s\sqrt{{D}_1{D}_2}&\hspace{-0.1in}{D}_2\hspace{-0.1in}\end{array}\right]: \alpha\in[0,1]\right\}\nonumber,
\end{eqnarray}
with $\tilde{\theta}$ defined in (\ref{thetastardef}) and $s\stackrel{\Delta}{=}\mathrm{sign}(\tilde{\theta})$.
\end{lemma}

For a convex optimization problem with objective function $f$, inequality constraints $g_i\le 0$ for $j=1,...,m$ and equality constraints $h_j=0$ for $j=1,...,l$, the global optimal point $\boldsymbol{x}=\boldsymbol{x}^*$ must satisfy
\begin{eqnarray}
\nonumber
\boldsymbol{0}&\in&\partial f(\boldsymbol{\boldsymbol{x}}^*)+\sum_{i=1}^m\mu_i\partial g_i(\boldsymbol{x}^*)+\sum_{j=1}^l\lambda_j\partial h_i(\boldsymbol{x}^*),\\
\nonumber
g_i(\boldsymbol{x}^*)&\le& 0,i=1,2,...,m,\\
\nonumber
h_j(\boldsymbol{x}^*)&=&0,j=1,2,...,l,\\
\nonumber
\mu_i&\ge& 0,i=1,2,...,m,\\
\nonumber
\mu_ig_i(\boldsymbol{x}^*)&=&0,i=1,2,...,m,
\end{eqnarray}
for some $\mu_i$'s and $\lambda_j$'s.

\subsection{A new sufficient condition for sum-rate tightness}

Now we are ready to state our main result on a new sufficient condition for the tightness of BT minimum sum-rate. Consider an MT source coding problem defined by $\Sigma_{Y_\mathcal{L}}$ and $D_\mathcal{L}$. Denote the BT minimum sum-rate as $R_{sum}^{BT}(\Sigma_{Y_\mathcal{L}},D_\mathcal{L})$, and assume that the optimal BT scheme achieves a distortion matrix $\tilde{\boldsymbol{D}}$.
The main result of this paper is given in the following theorem.

\begin{theorem}\label{thmmain}
$R_{sum}^{BT}(\Sigma_{Y_\mathcal{L}},D_\mathcal{L})=R_{sum}(\Sigma_{Y_\mathcal{L}},D_\mathcal{L})$ if there exists a permutation $\pi$, a $\pi^{(K)}$ block diagonal p.d. matrix $\Sigma_{N_\mathcal{L}}$ such that $\Sigma_{N_\mathcal{L}}\preceq\Sigma_{Y_\mathcal{L}}$, an $L\times L$ p.s.d. matrix $\boldsymbol{\Omega}$, an $L\times L$ p.s.d. diagonal matrix $\boldsymbol{\Pi}$, and a set of $K$ $2\times 2$ p.s.d. matrices $\boldsymbol{\Theta}_j$, $j\in\mathcal{K}$ such that the following conditions are satisfied:
\begin{eqnarray}
\tilde{\boldsymbol{D}}\Big(\boldsymbol{\Pi}-\tilde{\boldsymbol{D}}^{-1} +\tilde{\boldsymbol{D}}^{-1}(\tilde{\boldsymbol{D}}^{-1}+\Sigma_{N_\mathcal{L}}^{-1}-\Sigma_{Y_\mathcal{L}}^{-1})^{-1}\tilde{\boldsymbol{D}}^{-1}\Big)\tilde{\boldsymbol{D}} &=&\boldsymbol{\Lambda}-\boldsymbol{\Omega},\label{mainKKT1}\\
\langle\boldsymbol{\Lambda}\rangle^\pi_j+\boldsymbol{\Theta}_j-\bbdelta\left(\langle\Sigma_{N_\mathcal{L}}\rangle^\pi_j, \mathrm{diag}(\langle\tilde{\boldsymbol{\Gamma}}\rangle^\pi_j)\right) &\ni&\boldsymbol{0},\forall j\in\mathcal{K},\label{mainKKT2}\\
\big[{\boldsymbol{\Lambda}}\big]_{\pi_k,\pi_k}&=&\big[\tilde{\boldsymbol{\Gamma}}\big]_{\pi_k,\pi_k},~k=2K+1,...,L,\label{mainKKT3}\\
\boldsymbol{\Omega}(\Sigma_{Y_\mathcal{L}}^{-1}-\tilde{\boldsymbol{D}}^{-1})&=&\boldsymbol{0},\label{mainKKT4}\\
\boldsymbol{\Theta}_j(\langle\Sigma_{N_\mathcal{L}}\rangle^\pi_j-\langle\tilde{\boldsymbol{\Gamma}}\rangle^\pi_j)&=&\boldsymbol{0},\forall j\in\mathcal{K},\label{mainKKT5}\\
~[\boldsymbol{\Pi}]_{j,j}([\tilde{\boldsymbol{D}}]_{j,j}-D_j)&=&0,\forall j\in\mathcal{K},\label{mainKKT6}
\end{eqnarray}
where $\langle\boldsymbol{C}\rangle^\pi_j$ denotes the $2\times 2$ submatrix constructed from the $(\pi_{2j-1},\pi_{2j})$-th row and $(\pi_{2j-1},\pi_{2j})$-th column of $\boldsymbol{C}$, and
\begin{eqnarray}
\tilde{\boldsymbol{\Gamma}}&\stackrel{\Delta}{=}&\Big(\tilde{\boldsymbol{D}}^{-1}+\Sigma_{N_\mathcal{L}}^{-1}-\Sigma_{Y_\mathcal{L}}^{-1}\Big)^{-1}.\label{Gammadef}
\end{eqnarray}
\end{theorem}

\vspace{0.1in}
\begin{proof}
To prove Theorem \ref{thmmain}, we need the following two lemmas, whose proofs are given in Appendices C, and D, respectively.

\vspace{0.1in}
\begin{lemma}\label{lemmasamestruct}
For any random objects $\boldsymbol{Y}_\mathcal{L}$ and $\boldsymbol{X}_\mathcal{M}$, if
\begin{eqnarray}
\Big[\mathrm{cov}(\boldsymbol{Y}_\mathcal{L}|\boldsymbol{X}_\mathcal{M})\Big]_{i,j}&=&0\label{lemmasamestruct1}
\end{eqnarray}
for some $i,j\in\mathcal{L}$, then
\begin{eqnarray}
\Big[\mathrm{cov}(\boldsymbol{Y}_\mathcal{L}|\boldsymbol{X}_\mathcal{M},W_\mathcal{L})\Big]_{i,j}&=&0\label{lemmasamestruct2}
\end{eqnarray}
for any $L$ functions $W_\mathcal{L}\stackrel{\Delta}{=}\big\{\psi_1^{(n)}(\boldsymbol{Y}_1), \psi_2^{(n)}(\boldsymbol{Y}_2), ...,  \psi_L^{(n)}(\boldsymbol{Y}_L)\big\}$. \end{lemma}
\vspace{0.1in}

\begin{lemma}\label{lemmaconnect}
For any pair $({X_\mathcal{M}},{Y_\mathcal{L}})$ satisfying (\ref{correlationstruct}) and any $D_\mathcal{L}$, there exists a $\boldsymbol{D}\in\mathbb{R}^{L\times L}$ and a \begin{eqnarray}
\boldsymbol{\Gamma}&=&\bbpi^T\mathrm{diag}(\boldsymbol{\Gamma}_1, ..., \boldsymbol{\Gamma}_K, \gamma_{K+1},...,\gamma_L)\bbpi~\in~\Upsilon_K(\pi)\label{Gammastruct}
\end{eqnarray}

\vspace{-0.15in}
\noindent
such that
\vspace{-0.25in}

\begin{eqnarray}
\mathrm{diag}(\boldsymbol{D})&\le& D_\mathcal{L}\nonumber\\
\boldsymbol{\Gamma}&\preceq&(\boldsymbol{D}^{-1}+\Sigma_{N_\mathcal{L}}^{-1}-\Sigma_{Y_\mathcal{L}}^{-1})^{-1},
\end{eqnarray}
and the sum-rate of the quadratic Gaussian $L$-terminal problem satisfies
\begin{eqnarray}
&&R_{sum}(\Sigma_{Y_\mathcal{L}},D_\mathcal{L})\nonumber\\&\ge&
\frac{1}{2}\log\frac{|\Sigma_{X_\mathcal{M}}|}{|\boldsymbol{A}\boldsymbol{D}\boldsymbol{A}^T+\boldsymbol{B}|}
+\sum_{k=1}^K\underline{R}_{sum}(\Sigma_{Y_{\{\pi_{2k-1},\pi_{2k}\}}|{X}_\mathcal{M}},\boldsymbol{\Gamma}_k)+\frac{1}{2}\sum_{i=K+1}^L\log \frac{\sigma_{N_{\pi_i}}^2}{\gamma_i},\label{lemmaconnentstatement}
\end{eqnarray}
where $\Sigma_{Y_{\{\pi_{2k-1},\pi_{2k}\}}|{X}_\mathcal{M}}$ denotes the conditional covariance matrix of $(Y_{\pi_{2k-1}}, Y_{\pi_{2k}})^T$ given ${X}_\mathcal{M}$, and $\boldsymbol{A}$ and $\boldsymbol{B}$ are defined in (\ref{ABdef}).
\end{lemma}

{\bf Remarks:}
\begin{itemize}
\item Lemma \ref{lemmasamestruct} ensures that $\mathrm{cov}(\boldsymbol{Y}_\mathcal{L}|\boldsymbol{X}_\mathcal{M},W_\mathcal{L})$ in (\ref{lemmasamestruct2}) shares the same structure with $\Sigma_{N_\mathcal{L}}=\mathrm{cov}(\boldsymbol{Y}_\mathcal{L}|\boldsymbol{X}_\mathcal{M})$ in (\ref{lemmasamestruct1}), which is assumed to be block-diagonal in this paper. Note that this property holds even for non-block-diagonal $\Sigma_{N_\mathcal{L}}$'s.

\item This structural similarity between $\Sigma_{N_\mathcal{L}}=\mathrm{cov}(\boldsymbol{Y}_\mathcal{L}|\boldsymbol{X}_\mathcal{M})$ and $\mathrm{cov}(\boldsymbol{Y}_\mathcal{L}|\boldsymbol{X}_\mathcal{M},W_\mathcal{L})$ is a key to the proof of Lemma 5, since it restricts $\mathrm{cov}(\boldsymbol{Y}_\mathcal{L}|\boldsymbol{X}_\mathcal{M},W_\mathcal{L})$, which equals to $\boldsymbol{\Gamma}$ in (\ref{Gammastruct}), to be block-diagonal, and hence makes the lower bound (\ref{lemmaconnentstatement}) much simpler.
\end{itemize}

\vspace{0.1in}
Now we proceed to the proof of Theorem \ref{thmmain}.

Due to Lemma \ref{lemmaconnect}, to find the best lower bound on $R_{sum}(\Sigma_{Y_\mathcal{L}},D_\mathcal{L})$, we need to solve the following optimization problem for given  $({X_\mathcal{M}},{Y_\mathcal{L}})$ and $D_\mathcal{L}$ satisfying (\ref{correlationstruct}),
\begin{eqnarray}
\mathrm{Minimizing} &&\frac{1}{2}\log\frac{|\Sigma_{X_\mathcal{M}}|}{|\boldsymbol{A}\boldsymbol{D}\boldsymbol{A}^T+\boldsymbol{B}|}
+\sum_{k=1}^K\underline{R}_{sum}(\Sigma_{Y_{\{\pi_{2k-1},\pi_{2k}\}}|{X}_\mathcal{M}},\boldsymbol{\Gamma}_k)+\frac{1}{2}\sum_{i=K+1}^L\log \frac{\sigma_{N_{\pi_i}}^2}{\gamma_i}
\nonumber\\\mathrm{over~~}&&{\boldsymbol{D},\boldsymbol{\Gamma}_1,...,\boldsymbol{\Gamma}_K,\gamma_{2K+1},...,\gamma_L} \nonumber\\
\mathrm{subject~to}&&\boldsymbol{\Gamma}\preceq(\Sigma_{N_\mathcal{L}}^{-1}+\boldsymbol{D}^{-1}-\Sigma_{Y_\mathcal{L}}^{-1})^{-1},\nonumber\\
&&\boldsymbol{0}~\prec~\boldsymbol{D}~\preceq~\Sigma_{Y_\mathcal{L}},\nonumber\\
&&[\boldsymbol{D}]_{j,j}~\le~ D_j,~~~\mathrm{for~any~}j\in\mathcal{L},\nonumber\\
&&\boldsymbol{0}~\prec~\boldsymbol{\Gamma}_k~\preceq~\Sigma_{N_{\{\pi_{2k-1},\pi_{2k}\}}} \forall k\in\mathcal{K},\nonumber\\
&&0<\gamma_k\le\sigma_{N_{\pi_k}}^2,~k=2K+1,...,L,\nonumber
\end{eqnarray}
which is clearly convex. The Lagrangian is
{\begin{eqnarray}
\mathbbm{L}&=&-\frac{1}{2}\log|\boldsymbol{A}\boldsymbol{D}\boldsymbol{A}^T+\boldsymbol{B}|\nonumber\\&&  +\sum_{k=1}^K\underline{R}_{sum}(\Sigma_{Y_{\{\pi_{2k-1},\pi_{2k}\}}|{X}_\mathcal{M}},\boldsymbol{\Gamma}_k) -\frac{1}{2}\sum_{i=K+1}^L\log {\gamma_i}\nonumber\\&&+\mathrm{tr}\Big(\Lambda\big((\Sigma_{N_\mathcal{L}}^{-1}+\boldsymbol{D}^{-1}-\Sigma_{Y_\mathcal{L}}^{-1})-\boldsymbol{\Gamma^{-1}}\big)\Big)
+\mathrm{tr}\Big(\Omega(\Sigma_{Y_\mathcal{L}}^{-1}-{\boldsymbol{D}}^{-1})\Big) \nonumber\\&&+\sum_{i=1}^K \mathrm{tr}\Big(\Theta_i(\Sigma_{N_{\{\pi_{2i-1},\pi_{2i}\}}}^{-1}-\boldsymbol{\Gamma}_i^{-1})\Big)+\sum_{j=1}^L\mathrm{tr}(\Pi_j \boldsymbol{E}_j{\boldsymbol{D}}\boldsymbol{E}_j),\nonumber
\end{eqnarray}}
where $\Lambda$, $\Omega$, $\Theta_i$, $i\in\mathcal{K}$, $\Pi_j$, $j\in\mathcal{L}$ are p.s.d. matrices, and $\boldsymbol{E}_i$ is the $L\times L$ single-entry matrix whose $(i,i)$-th element is one.

Assume that the optimal BT scheme achieves a distortion matrix $\tilde{\boldsymbol{D}}$, and ${\tilde{\boldsymbol{\Gamma}}}$ as defined in (\ref{Gammadef}), then by applying Lemma \ref{lemmasubdiff}, we obtain the subgradient based KKT conditions at $(\tilde{\boldsymbol{D}},{\tilde{\boldsymbol{\Gamma}}})$, which are
\begin{eqnarray}
\tilde{\boldsymbol{D}}\Big(\boldsymbol{\Pi}-\tilde{\boldsymbol{D}}^{-1} +\tilde{\boldsymbol{D}}^{-1}(\tilde{\boldsymbol{D}}^{-1}+\Sigma_{N_\mathcal{L}}^{-1}-\Sigma_{Y_\mathcal{L}}^{-1})^{-1}\tilde{\boldsymbol{D}}^{-1}\Big)\tilde{\boldsymbol{D}} &=&\boldsymbol{\Lambda}-\boldsymbol{\Omega},\nonumber\\
\langle\boldsymbol{\Lambda}\rangle^\pi_j+\boldsymbol{\Theta}_j-\bbdelta\left(\langle\Sigma_{N_\mathcal{L}}\rangle^\pi_j, \mathrm{diag}(\langle\tilde{\boldsymbol{\Gamma}}\rangle^\pi_j)\right) &\ni&\boldsymbol{0},\forall j\in\mathcal{K},\nonumber\\
\big[{\boldsymbol{\Lambda}}\big]_{\pi_k,\pi_k}&=&\big[\tilde{\boldsymbol{\Gamma}}\big]_{\pi_k,\pi_k},~k=2K+1,...,L,\nonumber\\
\boldsymbol{\Omega}(\Sigma_{Y_\mathcal{L}}^{-1}-\tilde{\boldsymbol{D}}^{-1})&=&\boldsymbol{0},\nonumber\\
\boldsymbol{\Theta}_j(\langle\Sigma_{N_\mathcal{L}}\rangle^\pi_j-\langle\tilde{\boldsymbol{\Gamma}}\rangle^\pi_j)&=&\boldsymbol{0},\forall j\in\mathcal{K},\nonumber\\
~[\boldsymbol{\Pi}]_{j,j}([\tilde{\boldsymbol{D}}]_{j,j}-D_j)&=&0,\forall j\in\mathcal{K},\nonumber
\end{eqnarray}
where $\boldsymbol{\Pi}$, $\boldsymbol{\Lambda}$, $\boldsymbol{\Omega}$, and $\boldsymbol{\Theta}_j$'s are the p.s.d. Lagrangian multipliers. Then Theorem \ref{thmmain} readily follows.
\end{proof}

\vspace{0.05in}
\noindent$\bullet$\hspace{0.08in}{Example 1:} the block-degraded case

All known cases of quadratic Gaussian MT source coding problems with tight sum-rate bound belong to the non-degraded subclass, where all target distortions are met with equalities (i.e., all distortion constraints are active \cite{Boydbook}) in the optimal BT scheme. In this subsection, we first study a block-degraded case, and independently show sum-rate tightness in this case (under certain condition). Then we give a numerical example to confirm that the set of block-degraded case with tight sum-rate intersects with the one defined by the sufficient condition in Theorem \ref{thmmain}.

Consider a special case of quadratic Gaussian MT source coding,
where the vector source $Y_\mathcal{L}$ and the target distortion
vector $D_\mathcal{L}$ can be partitioned into $K$ groups, namely,
$(Y_{\mathcal{S}_1},D_{\mathcal{S}_1})$,
$(Y_{\mathcal{S}_2},D_{\mathcal{S}_2}), ...,
(Y_{\mathcal{S}_K},D_{\mathcal{S}_K})$, and for any $k\in\mathcal{K}$,
there exists an integer $\mathbbm{i}(k)\in\mathcal{S}_k$, such that
\begin{eqnarray}
Y_{j}~=~Y_{\mathbbm{i}(k)}+Z_j,\mathrm{~~and~~}D_j
\ge D_{\mathbbm{i}(k)}+\sigma_{Z_j}^2,\forall j\in\mathcal{S}_k,
\end{eqnarray}
where $Z_j~\sim\mathcal{N}(0,\sigma_{Z_j}^2)$ with
$\sigma_{Z_j}^2>0$ for $j\in\mathcal{S}_k-\{\mathbbm{i}(k)\}$ and
$\sigma_{Z_{\mathbbm{i}(k)}}^2=0$ is independent of
$Y_{\mathbbm{i}(k)}$ and $Z_j$'s are mutually independent.
Each $Y_{\mathbbm{i}(k)}$, $k\in\mathcal{K}$ is called the
{\em group leader} in $Y_{\mathcal{S}_k}$, and denote
$\bar{Y}_\mathcal{K}=(Y_{\mathbbm{i}(1)},Y_{\mathbbm{i}(2)},...,
Y_{\mathbbm{i}(k)})^T$, $\bar{D}_\mathcal{K}=(D_{\mathbbm{i}(1)},
D_{\mathbbm{i}(2)},...,D_{\mathbbm{i}(k)})^T$.
We say a pair $(\Sigma_{Y_\mathcal{L}},D_\mathcal{L})$ is
{\em block-degraded (BD)} if they satisfy the above condition.
The $K$ components of $\bar{Y}_\mathcal{K}$ are referred to as
{\em core sources} while the other $L-K$ as {\em redundant sources}.

Equivalently, $(\Sigma_{Y_\mathcal{L}},D_\mathcal{L})$ is BD if there
exists a partition $\mathcal{P}=\{\mathcal{S}_k:k\in\mathcal{K}\}$
of $\mathcal{L}$ and another pair
$(\Sigma_{\bar{Y}_\mathcal{K}},\bar{D}_\mathcal{K})$ such that
\begin{eqnarray}
\hspace{-0.15in}\Sigma_{Y_\mathcal{L}}
&=&\boldsymbol{G}_\mathcal{P}\Sigma_{\bar{Y}_\mathcal{K}}
\boldsymbol{G}_\mathcal{P}^T+\Sigma_{Z_\mathcal{L}},\label{BDcond1}\\
\hspace{-0.15in}D_{\mathbbm{i}(k)}&=&\bar{D}_k,\forall k\in\mathcal{K},
\label{BDcond2}\\
\hspace{-0.15in}D_j&\ge&\bar{D}_k+[\Sigma_{Z_\mathcal{L}}]_{j,j},
\forall j\in\mathcal{S}_k-\{\mathbbm{i}(k)\}\mathrm{~and~}k\in\mathcal{K},
\label{BDcond3}
\end{eqnarray}
where $\boldsymbol{G}_\mathcal{P}$ is an $L\times K$ matrix whose
$(j,\mathbbm{i}(k))$-th element is one for all $j\in\mathcal{S}_k$,
$k\in\mathcal{K}$ with the rest being zero, and
$\Sigma_{Z_\mathcal{L}}$ is a diagonal matrix whose diagonal elements
are positive with exceptions that
$[\Sigma_{Z_{\mathcal{L}}}]_{\mathbbm{i}(k),\mathbbm{i}(k)}=0$.
Then an $L$-terminal quadratic Gaussian MT source coding problem with a
BD pair $(\Sigma_{Y_\mathcal{L}},D_\mathcal{L})$ automatically induces
a $K$-terminal source coding problem defined by the
pair $(\Sigma_{\bar{Y}_\mathcal{K}},\bar{D}_\mathcal{K})$.

Consider a BD pair $(\Sigma_{Y_\mathcal{L}},D_\mathcal{L})$
with partition $\mathcal{P}=\{\mathcal{S}_k:k\in\mathcal{K}\}$
and $(\Sigma_{\bar{Y}_\mathcal{K}},\bar{D}_\mathcal{K},
\Sigma_{Z_\mathcal{L}})$ satisfying (\ref{BDcond1})-(\ref{BDcond3}).
We say a matrix $\Lambda$ is $\mathcal{P}$-block-diagonal if
$[\Lambda]_{i,j}=0$ for any $i\in\mathcal{S}_k,j\in\mathcal{S}_l$
with $k,l\in\mathcal{K}, k\neq l$, and denote $\mathbbm{d}_\mathcal{P}$
as the set of all $\mathcal{P}$-block-diagonal matrices.
For two $L\times L$ matrices $A$ and $B$, we write
$A\stackrel{\mathcal{P}}{\equiv}B$ if $[A]_{i,j}=[B]_{i,j}$
for any $i,j\in\mathcal{S}_k$ with some $k\in\mathcal{K}$.

We claim that for a BD pair $(\Sigma_{Y_\mathcal{L}},D_\mathcal{L})$, tightness of the BT sum-rate bound in the induced $K$-terminal quadratic
Gaussian MT source coding problem implies tightness of the same bound
in the original $L$-terminal problem, which is stated in the following lemma, whose proof is given in Appendix E.

\begin{lemma}\label{lemmaBD}
For any BD pair $(\Sigma_{Y_\mathcal{L}},D_\mathcal{L})$,
if the BT minimum sum-rate is tight for the induced $K$-terminal
source coding problem, i.e.,
\begin{eqnarray}
R_{sum}(\Sigma_{\bar{Y}_\mathcal{K}},\bar{D}_\mathcal{K})
&=&R_{sum}^{BT}(\Sigma_{\bar{Y}_\mathcal{K}},\bar{D}_\mathcal{K}),
\end{eqnarray}
then it must also be tight for the original MT source coding
problem defined by $(\Sigma_{Y_\mathcal{L}},D_\mathcal{L})$, i.e.,
\begin{eqnarray}
R_{sum}(\Sigma_{{Y}_\mathcal{L}},{D}_\mathcal{L})
&=&R^{BT}_{sum}(\Sigma_{{Y}_\mathcal{L}},{D}_\mathcal{L})
~=~R_{sum}^{BT}(\Sigma_{\bar{Y}_\mathcal{K}},\bar{D}_\mathcal{K}).\nonumber
\end{eqnarray}
\end{lemma}

{\bf Remarks:}
\begin{itemize}
\item Although Wang {\em et al.}'s sufficient condition \cite{WangISIT09} for sum-rate tightness does not include any degraded case, one can easily use Lemma \ref{lemmaBD} to generate a BD example with tight sum-rate bound. In fact, with slight modifications (with details omitted), Wang {\em et al.}'s proof \cite{WangISIT09} can also be generalized to directly show sum-rate tightness for such BD cases without explicitly using Lemma \ref{lemmaBD}.
\item We note that Lemma \ref{lemmaBD} only guarantees the sum-rate tightness of a {\em subset} of the BD subclass of quadratic Gaussian MT source coding problems. Moreover, this subset intersects with the one defined by the sufficient condtion in Theorem \ref{thmmain}, as shown in the following numerical example.

\end{itemize}

A specific numerical example that satisfies the requirements in both Theorem \ref{thmmain} and Lemma \ref{lemmaBD} is as follows. Let $L=4$,
\begin{eqnarray}\Sigma_{Y_\mathcal{L}} &=& \left[
\begin{array}{rrrrrr}
    1.0000  &  0.9000 &   0.8000 & 0.8000\\
    0.9000  &  1.0000 &   0.7000  & 0.7000\\
    0.8000  &  0.7000 &   1.0000 & 1.0000\\
    0.8000  &  0.7000 &   1.0000 & 1.1000
\end{array}\right],
\end{eqnarray}
and
\begin{eqnarray}
D_\mathcal{L}&=&(0.3760, 0.35, 0.3, 0.5)^T,
\end{eqnarray}
The optimal BT distortion matrix is
\begin{eqnarray}\tilde{\boldsymbol{D}} &=& \left[
\begin{array}{rrrrrr}
    0.3760  &  0.2740  &  0.1818 &   0.1818\\
    0.2740  &  0.3500  &  0.1231 &   0.1231\\
    0.1818  &  0.1231  &  0.3000 &   0.3000\\
    0.1818  &  0.1231  &  0.3000 &   0.4000
\end{array}\right],
\end{eqnarray}
hence this example is degraded since $D_4=0.5$ is not achieved with equality in the optimal BT distortion matrix $\tilde{\boldsymbol{D}}$.

We first verify that this example satisfies the sufficient condition in Theorem \ref{thmmain}. Let $\pi=\{1,2,3,4\}$ and
\begin{eqnarray}\Sigma_{N_\mathcal{L}} &=& \left[
\begin{array}{rrrrr}
    0.2942  &  0.2852  &       0 &0\\
    0.2852  &  0.4535  &       0 &0\\
         0  &       0  &  0.0923 &0\\
         0  &       0  &  0&0.1923
\end{array}\right]\label{eqn69}
\end{eqnarray}
be a $\pi^{(K)}$ p.d. block diagonal matrix with $K=1$. Then $M=4$,
\begin{eqnarray}
\Sigma_{X_\mathcal{M}} &=& \left[
\begin{array}{rrrrr}
    3.1162    &     0  &       0  &       0\\
         0    &0.0923  &       0  &       0\\
         0    &     0  &  0.0377  &       0\\
         0    &     0  &       0  &  0.0061
\end{array}\right],\nonumber\\\boldsymbol{H} &=& \left[
\begin{array}{rrrrr}
   -0.4712  & -0.4130 &  -0.5511  & -0.5511\\
   0  & 0 &   0.7071  & -0.7071\\
    0.5417  &  0.5619 &  -0.4421  & -0.4421\\
   -0.6961  &  0.7167 &   0.0290  &  0.0290
\end{array}\right].
\end{eqnarray}
Now the following p.s.d. matrices
\begin{eqnarray}
{\boldsymbol{\Lambda}} &=& \left[
\begin{array}{rrrrr}
    0.2248  &  0.2489  &  0.0967 &   0.0967\\
    0.2489  &  0.2791  &  0.1075 &   0.1075\\
    0.0967  &  0.1075  &  0.0783 &   0\\
    0.0967  &  0.1075  &  0 &   0.1923
\end{array}\right],\\
{\boldsymbol{\Omega}} &=& \left[
\begin{array}{rrrrr}
         0    &     0     &    0    &     0\\
         0    &     0     &    0    &     0\\
         0    &     0     &    0    &     0\\
         0    &     0     &    0    &     0.1000
\end{array}\right],\\
{\boldsymbol{\Theta_1}} &=& \left[
\begin{array}{rrrrr}
         0    &     0\\
         0    &     0
\end{array}\right],\\
{\boldsymbol{\Pi}} &=& \left[
\begin{array}{rrrrr}
          1.0377     &     0    &     0 &        0\\
         0    &     1.8957      &   0  &       0\\
         0    &     0        & 2.6331 &        0\\
         0    &     0        & 0 &        0
\end{array}\right],\\
\tilde{\boldsymbol{\Gamma}} &=& \left[
\begin{array}{rrrrr}
    0.2248  &  0.1753  &  0 &0\\
    0.1753  &  0.2791  &0   & 0\\
    0  & 0   & 0.0783  & 0\\
    0  &  0  & 0  &  0.1923
\end{array}\right]\label{eqn75}
\end{eqnarray}
satisfy all the KKT conditions. Note that $\tilde{\boldsymbol{\Gamma}}$ in (\ref{eqn75}) has the same structure as $\Sigma_{N_\mathcal{L}}$ in (\ref{eqn69}), which is consistent with Lemma \ref{lemmasamestruct}. In addition, $\bbdelta(\Sigma_{Y_\mathcal{L}},D_\mathcal{L})$ is a line segment
\begin{eqnarray}
\bbdelta(\Sigma_{Y_\mathcal{L}},D_\mathcal{L})\Big\}&=&\left\{\alpha\cdot\left[
\begin{array}{rrrrr}   0.2248 &  0.2505\\   0.2505 &  0.2791\end{array}\right]+(1-\alpha)\cdot\left[\begin{array}{rrrrr}
   0.2248  & 0.1001\\   0.1001 &  0.2791\end{array}\right]:\alpha\in[0,1]\right\}.
\end{eqnarray}

On the other hand, it is easy to verify that $(\Sigma_{Y_\mathcal{L}},D_\mathcal{L})$ is a BD pair with
\begin{eqnarray}
\mathcal{P}&=&\Big\{\{1\},\{2\},\{3,4\}\Big\},\nonumber\\
\Sigma_{\bar{Y}_\mathcal{K}} &=& \left[
\begin{array}{rrr}
    1.0000  &  0.9000 &   0.8000 \\
    0.9000  &  1.0000 &   0.7000  \\
    0.8000  &  0.7000 &   1.0000
\end{array}\right],\nonumber\\
\Sigma_{Z_\mathcal{L}}&=&\mathrm{diag}(0, 0, 0, 0.1),\nonumber\\
\bar{D}_\mathcal{K}&=&(0.3760, 0.35, 0.3)^T,\nonumber
\end{eqnarray}
and the induced three-terminal quadratic Gaussian MT source coding problem defined by $(\Sigma_{\bar{Y}_\mathcal{K}},\bar{D}_\mathcal{L})$ has a tight sum-rate bound due to Theorem \ref{thmmain}. Hence we conclude that the above four-terminal numerical example of quadratic Gaussian MT source coding problem also satisfies the simple sufficient condition in Lemma \ref{lemmaBD}.

\section{A simplified sufficient condition}

Although the sufficient condition given in Theorem \ref{thmmain} is more inclusive than that in \cite{WangISIT09}, it is rather complicated and hard to verify. However, in the {\em non-degraded} case where the optimal BT scheme quantizes every source, and achieves all $L$ target distortions with equalities, the sufficient condition in Theorem 2 can be further simplified. Note that the non-degraded case is of special interest since all the previously known quadratic Gaussian MT source coding problems with tight sum-rate bound belong to this case.

\begin{corollary}\label{corollarymain}
For an MT source coding problem defined by $\Sigma_{Y_\mathcal{L}}$ and $D_\mathcal{L}$, if the optimal BT distortion matrix $\tilde{\boldsymbol{D}}$ satisfies $\mathrm{diag}(\tilde{\boldsymbol{D}})=D_\mathcal{L}$ and $\tilde{\boldsymbol{D}}^{-1}-\Sigma_{Y_\mathcal{L}}^{-1}$ is a p.d. matrix, then $R_{sum}^{BT}(\Sigma_{Y_\mathcal{L}},D_\mathcal{L})=R_{sum}(\Sigma_{Y_\mathcal{L}},D_\mathcal{L})$ if there exists a permutation $\pi$ and a $\pi^{(K)}$ block diagonal p.d. matrix $\Sigma_{N_\mathcal{L}}$ such that $\Sigma_{N_\mathcal{L}}\preceq\Sigma_{Y_\mathcal{L}}$,
\begin{eqnarray}
{\boldsymbol{\Lambda}}\stackrel{\Delta}{=}\tilde{\boldsymbol{D}}\Big({\boldsymbol{\Pi}}-\tilde{\boldsymbol{D}}^{-1} +\tilde{\boldsymbol{D}}^{-1}(\tilde{\boldsymbol{D}}^{-1}+\Sigma_{N_\mathcal{L}}^{-1}-\Sigma_{Y_\mathcal{L}}^{-1})^{-1} \tilde{\boldsymbol{D}}^{-1}\Big)\tilde{\boldsymbol{D}}\label{Lambdacorollary}
\end{eqnarray}
is a p.s.d. matrix, and
\begin{eqnarray}
\mathrm{~sign}(\big[{\tilde{\boldsymbol{\Gamma}}}\big]_{\pi_{2k-1},\pi_{2k-1}}) \cdot\big[\boldsymbol{\Lambda}\big]_{\pi_{2k-1},\pi_{2k}}&\ge& 2\Big|\big[{\boldsymbol{\Gamma}}\big]_{\pi_{2k-1},\pi_{2k}}\Big| -\sqrt{\big[{\boldsymbol{\Gamma}}\big]_{\pi_{2k-1},\pi_{2k-1}}\big[{\boldsymbol{\Gamma}}\big]_{\pi_{2k},\pi_{2k}}}\label{corcond}
\end{eqnarray}
is satisfied for all $k\in\mathcal{K}$, where $\tilde{\boldsymbol{\Gamma}}$ is defined in (\ref{Gammadef}) and
\begin{eqnarray}
{\boldsymbol{\Pi}}\stackrel{\Delta}{=}\mathrm{diag}\Big((\tilde{\boldsymbol{D}}\odot\tilde{\boldsymbol{D}})^{-1}D_\mathcal{L}\Big),\label{Pidef}
\end{eqnarray}
with $\odot$ denoting Hadamard product (entrywise product).
\end{corollary}

\begin{proof}
First, due to the assumption that $\tilde{\boldsymbol{D}}^{-1}-\Sigma_{Y_\mathcal{L}}^{-1}\succ\boldsymbol{0}$, (\ref{mainKKT4}) implies that $\boldsymbol{\Omega}=\boldsymbol{0}$, which, combined with (\ref{mainKKT1}), directly leads to (\ref{Lambdacorollary}). On the other hand, $\tilde{\boldsymbol{D}}^{-1}-\Sigma_{Y_\mathcal{L}}^{-1}\succ\boldsymbol{0}$ also ensures that
\begin{eqnarray}
\tilde{\boldsymbol{\Gamma}}&=&\Big(\tilde{\boldsymbol{D}}^{-1}+\Sigma_{N_\mathcal{L}}^{-1}-\Sigma_{Y_\mathcal{L}}^{-1}\Big)^{-1}~\prec~\Sigma_{N_\mathcal{L}},
\end{eqnarray}
hence (\ref{mainKKT5}) is true if and only if $\boldsymbol{\Theta}_j=\boldsymbol{0}$ for all $j\in\mathcal{K}$.

Now (\ref{mainKKT2}) becomes
\begin{eqnarray}
\langle\boldsymbol{\Lambda}\rangle^\pi_j-\bbdelta\left(\langle\Sigma_{N_\mathcal{L}}\rangle^\pi_j, \mathrm{diag}(\langle\tilde{\boldsymbol{\Gamma}}\rangle^\pi_j)\right)&\ni&\boldsymbol{0},\forall j\in\mathcal{K},
\end{eqnarray}
then due to the fact that all $2\times 2$ matrices in $\bbdelta\left(\langle\Sigma_{N_\mathcal{L}}\rangle^\pi_j, \mathrm{diag}(\langle\tilde{\boldsymbol{\Gamma}}\rangle^\pi_j)\right)$ have the same diagonal elements as those of $\langle\tilde{\boldsymbol{\Gamma}}\rangle^\pi_j$,
we know that
\begin{eqnarray}
\big[{\boldsymbol{\Lambda}}\big]_{\pi_k,\pi_k}&=&\big[\tilde{\boldsymbol{\Gamma}}\big]_{\pi_k,\pi_k},~\forall~k=1,2,...,2K.\label{mainKKT3else}
\end{eqnarray}
Hence by combining (\ref{mainKKT3}) and (\ref{mainKKT3else}), we obtain
\begin{eqnarray}
&\mathrm{diag}(\boldsymbol{\Lambda})&=~\mathrm{diag}(\tilde{\boldsymbol{\Gamma}})\nonumber\\
\Leftrightarrow&\mathrm{diag}(\tilde{\boldsymbol{D}}\Big({\boldsymbol{\Pi}}-\tilde{\boldsymbol{D}}^{-1} +\tilde{\boldsymbol{D}}^{-1}(\tilde{\boldsymbol{D}}^{-1}+\Sigma_{N_\mathcal{L}}^{-1}-\Sigma_{Y_\mathcal{L}}^{-1})^{-1}\tilde{\boldsymbol{D}}^{-1}\Big)\tilde{\boldsymbol{D}})&=~ \mathrm{diag}(\Big(\tilde{\boldsymbol{D}}^{-1}+\Sigma_{N_\mathcal{L}}^{-1}-\Sigma_{Y_\mathcal{L}}^{-1}\Big)^{-1})\nonumber\\
\Leftrightarrow&\mathrm{diag}(\tilde{\boldsymbol{D}}{\boldsymbol{\Pi}}\tilde{\boldsymbol{D}})&=~ \mathrm{diag}(\tilde{\boldsymbol{D}})\nonumber\\
\Leftrightarrow&\sum_{j=1}^L[\tilde{\boldsymbol{D}}]_{i,j}^2\cdot[\boldsymbol{\Pi}]_{j,j}&=~[\tilde{\boldsymbol{D}}]_{i,i},~\forall~i\in\mathcal{L}\nonumber\\
\Leftrightarrow&(\tilde{\boldsymbol{D}}\odot \tilde{\boldsymbol{D}})\mathrm{diag}(\boldsymbol{\Pi})&=\mathrm{diag}(\tilde{\boldsymbol{D}})~=~D_\mathcal{L}\nonumber\\
\Leftrightarrow&\mathrm{diag}(\boldsymbol{\Pi})&=~(\tilde{\boldsymbol{D}}\odot \tilde{\boldsymbol{D}})^{-1}D_\mathcal{L},
\end{eqnarray}
and (\ref{Pidef}) is proved.

Finally, (\ref{mainKKT3else}) holds if there exists an $\alpha\in[0,1]$ such that
\begin{eqnarray}
\big[\boldsymbol{\Lambda}\big]_{\pi_{2k-1},\pi_{2k}}&=&\left(\alpha +(1-\alpha)(2{\Big|\frac{\big[{\tilde{\boldsymbol{\Gamma}}}\big]_{\pi_{2k-1},\pi_{2k}}} {\sqrt{\big[{\tilde{\boldsymbol{\Gamma}}}\big]_{\pi_{2k-1},\pi_{2k-1}}\big[{\tilde{\boldsymbol{\Gamma}}}\big]_{\pi_{2k},\pi_{2k}}}}\Big|}-1)\right) \nonumber\\
&&\cdot\mathrm{~sign}(\big[{\tilde{\boldsymbol{\Gamma}}}\big]_{\pi_{2k-1},\pi_{2k-1}}) \sqrt{\big[{\tilde{\boldsymbol{\Gamma}}}\big]_{\pi_{2k-1},\pi_{2k-1}}\big[{\tilde{\boldsymbol{\Gamma}}}\big]_{\pi_{2k},\pi_{2k}}}.\label{corcond1}
\end{eqnarray}
Now (\ref{corcond1}) is equivalent to
\begin{eqnarray}
\mathrm{~sign}(\big[{\tilde{\boldsymbol{\Gamma}}}\big]_{\pi_{2k-1},\pi_{2k-1}}) \cdot\big[\boldsymbol{\Lambda}\big]_{\pi_{2k-1},\pi_{2k}}&\le& \sqrt{\big[{\tilde{\boldsymbol{\Gamma}}}\big]_{\pi_{2k-1},\pi_{2k-1}}\big[{\tilde{\boldsymbol{\Gamma}}}\big]_{\pi_{2k},\pi_{2k}}}\label{corcond21}\\
\mathrm{and~}\mathrm{~sign}(\big[{\tilde{\boldsymbol{\Gamma}}}\big]_{\pi_{2k-1},\pi_{2k-1}}) \cdot\big[\boldsymbol{\Lambda}\big]_{\pi_{2k-1},\pi_{2k}}&\ge&2\Big|\big[{\tilde{\boldsymbol{\Gamma}}}\big]_{\pi_{2k-1},\pi_{2k}}\Big|- \sqrt{\big[{\tilde{\boldsymbol{\Gamma}}}\big]_{\pi_{2k-1},\pi_{2k-1}}\big[{\tilde{\boldsymbol{\Gamma}}}\big]_{\pi_{2k},\pi_{2k}}}\label{corcond22},
\end{eqnarray}
where (\ref{corcond21}) is automatically satisfied since
\begin{eqnarray}
\big[{{\boldsymbol{\Lambda}}}\big]_{\pi_{2k-1},\pi_{2k-1}}=\big[{\tilde{\boldsymbol{\Gamma}}}\big]_{\pi_{2k-1},\pi_{2k-1}},~ \big[{{\boldsymbol{\Lambda}}}\big]_{\pi_{2k},\pi_{2k}}=\big[{\tilde{\boldsymbol{\Gamma}}}\big]_{\pi_{2k},\pi_{2k}},~\mathrm{and~} \langle\boldsymbol{\Lambda}\rangle^\pi_j\succeq\boldsymbol{0}.
\end{eqnarray}
Hence (\ref{corcond}) must hold.
\end{proof}

\vspace{0.05in}
\noindent$\bullet$\hspace{0.08in}{Example 2:} the block-circulant case

We study a special class of quadratic Gaussian MT source coding problem named the block-circulant case.

Let $L=2m$ be an even number, and assume that the source covariance matrix $\Sigma_{Y_\mathcal{L}}$ is {\em block-circulant}, i.e., it is of the form
\begin{eqnarray}
\Sigma_{Y_\mathcal{L}}&=&
\left[\begin{array}{ccccc}
B_1&B_2&B_3&...&B_m\\
B_m&B_1&B_2&...&B_{m-1}\\
...&...&...&...&...\\
B_2&B_3&B_4&...&B_1\\\end{array}\right],\nonumber
\end{eqnarray}
where $B_i=B_{m+2-i}$ for $i=2,3,...,m$ are p.d. symmetric $2\times 2$ blocks of the form
\begin{eqnarray}
B_i&=&\left[\begin{array}{cc}b_{i,1} &b_{i,2}\\b_{i,2} &b_{i,1}\end{array}\right].
\end{eqnarray}
Denote $\mathbbm{C}_L$ as the set of all $L\times L$ block-circulant matrices. We state several important properties of block-circulant matrices.
\begin{itemize}
\item Any $\Sigma\in\mathbbm{C}_L$ can be diagonalized by
\begin{eqnarray}
\boldsymbol{G}_L&\stackrel{\Delta}{=}&\boldsymbol{F}_m\otimes\boldsymbol{F}_2,
\end{eqnarray}
with $\otimes$ denoting Kronecker product, and $\boldsymbol{F}_m$ being the $m\times m$ {\em real Fourier matrix} \cite{Allerton:09} (which is orthogonal with
$\boldsymbol{F}_m\boldsymbol{F}_m^T=\boldsymbol{I}_m$). For example, when $L=6$,
\begin{eqnarray}
\boldsymbol{G}_6&=&\boldsymbol{F}_3\otimes\boldsymbol{F}_2~=~\left[\begin{array}{rrrrrr}
    0.4082  &   0.4082  &      0  &     0  &   0.5774  &   0.5774\\
    0.4082  &  -0.4082  &     0  &     0  &   0.5774  &  -0.5774\\
    0.4082  &   0.4082  &   0.5000  &   0.5000  &  -0.2887  &  -0.2887\\
    0.4082  &  -0.4082  &   0.5000  &  -0.5000  &  -0.2887  &   0.2887\\
    0.4082  &   0.4082  &  -0.5000  &  -0.5000  &  -0.2887  &  -0.2887\\
    0.4082  &  -0.4082  &  -0.5000  &   0.5000  &  -0.2887  &   0.2887
    \end{array}\right].
\end{eqnarray}

\item $\mathbbm{C}_L$ is a ring under matrix addition and multiplication. In particular, $\mathbbm{C}_L$ is closed under the following operation
\begin{eqnarray}
\boldsymbol{A}\star\boldsymbol{B}&\stackrel{\Delta}{=}&\boldsymbol{A}-\boldsymbol{A}(\boldsymbol{A}+\boldsymbol{B})^{-1}\boldsymbol{A} ~=~\boldsymbol{B}-\boldsymbol{B}(\boldsymbol{A}+\boldsymbol{B})^{-1}\boldsymbol{B}~\in~\mathbbm{C}_L, \mathrm{~for~any~}\boldsymbol{A},\boldsymbol{B}\in\mathbbm{C}_L.
\end{eqnarray}

\item For any $\boldsymbol{A}\in\mathbbm{C}_L$, there are $2\cdot\lceil\frac{L+1}{2}\rceil$ degrees of freedom in the $L$ eigenvalues of $\boldsymbol{A}$, with $\lceil x\rceil$ denoting the smallest integer larger than $x$.
\end{itemize}

We say a quadratic Gaussian MT source coding problem belongs to the {\em block-circulant} case if the source covariance matrix is block-circulant and all the target distortions are equal, i.e., $\Sigma_{Y_\mathcal{L}}\in\mathbbm{C}_L$ and $D_\mathcal{L}=D\cdot\boldsymbol{1}$. An important fact for this special case, which follows directly from the properties of block-circulant matrices, is that the optimal BT distortion matrix can be expressed analytically with
\begin{eqnarray}
\tilde{\boldsymbol{D}}&=&\Sigma_{Y_\mathcal{L}}\star q\boldsymbol{I}_L,\label{DBTdefBC}
\end{eqnarray}
where $q$ satisfies
\begin{eqnarray}
\sum_{i=1}^L \frac{1}{\frac{1}{\lambda_i}+\frac{1}{q}}&=&LD,
\end{eqnarray}
with $\lambda_i$, $i\in\mathcal{L}$ being the $L$ eigenvalues of $\Sigma_{Y_\mathcal{L}}$.

Now we are ready to investigate the tightness condition provided by Wang {\em et al.} \cite{WangISIT09} for this block-circulant case, which is given in the following lemma, the proof is detailed in Appendix F.

\begin{lemma}\label{lemmaWangconditionBC}
For any block-circulant quadratic Gaussian MT source coding problem, Wang {\em et al.}'s tightness condition \cite[Lemma 4]{WangISIT09} for the sum-rate bound to be tight is equivalent to
\begin{eqnarray}
\mathrm{diag}\Big((\tilde{\boldsymbol{D}}\odot\tilde{\boldsymbol{D}})^{-1}D\boldsymbol{1}\Big)&\succeq& \tilde{\boldsymbol{D}}^{-1}-\tilde{\boldsymbol{D}}^{-1}(\tilde{\boldsymbol{D}}^{-1}+\lambda_{min}^{-1}\boldsymbol{I}_L -\Sigma_{Y_\mathcal{L}}^{-1})^{-1}\tilde{\boldsymbol{D}}^{-1},\label{WangcondBC}
\end{eqnarray}
with $\tilde{\boldsymbol{D}}$ defined in (\ref{DBTdefBC}) and $\lambda_{min}$ being the smallest eigenvalue of $\Sigma_{Y_\mathcal{L}}$.
\end{lemma}

With Lemma \ref{lemmaWangconditionBC}, one can easily test whether Wang {\em et al.}'s tightness condition is satisfied by a block-circulant case of quadratic Gaussian MT source coding problem. For example, let $L=4$ and
\begin{eqnarray}\Sigma_{Y_\mathcal{L}} = \left[
\begin{array}{rrrrrr}
    1.0000  &  0.5000  &  0.9750  &  0.4800\\
    0.5000  &  1.0000  &  0.4800  &  0.9750\\
    0.9750  &  0.4800  &  1.0000  &  0.5000\\
    0.4800  &  0.9750  &  0.5000  &  1.0000
\end{array}\right]~\in~\mathbbm{C}_4,
\end{eqnarray}
and $D_\mathcal{L}=0.1362\cdot\boldsymbol{1}$.
Then the optimal BT distortion matrix is
\begin{eqnarray}\tilde{\boldsymbol{D}} = \left[
\begin{array}{rrrrrr}
    0.1362  &  0.0189  &  0.1142  &  0.0018\\
    0.0189  &  0.1362  &  0.0018  &  0.1142\\
    0.1142  &  0.0018  &  0.1362  &  0.0189\\
    0.0018  &  0.1142  &  0.0189  &  0.1362
\end{array}\right].
\end{eqnarray}

We first use Lemma \ref{lemmaWangconditionBC} to test Wang {\em et al.}'s tightness condition, which is not satisfied since
\begin{eqnarray}
\mathrm{diag}\Big((\tilde{\boldsymbol{D}}\odot\tilde{\boldsymbol{D}})^{-1}D\boldsymbol{1}\Big)&=&4.1631\boldsymbol{I}_4\nonumber\\&\nsucceq & \left[
\begin{array}{rrrrrr}
    7.5599  &  5.4290  & -3.6183  & -5.7492\\
    5.4290  &  7.5599  & -5.7492  & -3.6183\\
   -3.6183  & -5.7492  &  7.5599  &  5.4290\\
   -5.7492  & -3.6183  &  5.4290  &  7.5599
\end{array}\right]\nonumber\\&=&\tilde{\boldsymbol{D}}^{-1}-\tilde{\boldsymbol{D}}^{-1}(\tilde{\boldsymbol{D}}^{-1}+\lambda_{min}^{-1}\boldsymbol{I}_L -\Sigma_{Y_\mathcal{L}}^{-1})^{-1}\tilde{\boldsymbol{D}}^{-1}.
\end{eqnarray}

However, it is easy to verify that this example does satisfy the condition given in Corollary \ref{corollarymain}, since when $\pi=\{1,2,3,4\}$ and
\begin{eqnarray}\Sigma_{N_\mathcal{L}} = \left[
\begin{array}{rrrrr}
    0.0250  &  0.0200  &     0  &      0\\
    0.0200  &  0.0250  &     0  &      0\\
         0  &     0  &  0.0250  &  0.0200\\
         0  &     0  &  0.0200  &  0.0250
\end{array}\right]~\in~\Upsilon_2(\pi),
\end{eqnarray}
$\tilde{\boldsymbol{\Gamma}}$ and ${\boldsymbol{\Lambda}}$ defined in (\ref{Gammadef}) and (\ref{Lambdacorollary}) satisfy
\begin{eqnarray}
\mathrm{sign}([\tilde{\boldsymbol{\Gamma}}]_{2k-1,2k})\cdot[{\boldsymbol{\Lambda}}]_{2k-1,2k}&=&0.0219 \nonumber\\&\ge&0.0171~=~2[\tilde{\boldsymbol{\Gamma}}]_{2k-1,2k}-\sqrt{[\tilde{\boldsymbol{\Gamma}}]_{2k-1,2k-1}[\tilde{\boldsymbol{\Gamma}}]_{2k,2k}}, ~k=1,2.
\end{eqnarray}

{\bf Remarks:}
\begin{itemize}
\item Unlike the known cases with tight sum-rate bound including the two-terminal case \cite{Wagner05}, the positive-symmetric case \cite{Wagner05}, and the BEEV-ED case \cite{Allerton:09}, some of the block-circulant cases might not have a tight sum-rate bound if they do not satisfy the requirements in Corollary \ref{corollarymain}.
\item We pick the block-circulant case as an example mainly because of the nice properties in this case that enable us to analytically evaluate the sufficient condition in Theorem \ref{thmWang} without a full search over $\Sigma_{N_\mathcal{L}}\in\mathscr{N}(\Sigma_{Y_\mathcal{L}})$.
\end{itemize}

\vspace{0.05in}
\noindent$\bullet$\hspace{0.08in}{Example 3:} another numerical example

Now we give a general numerical example that satisfies the requirement of Corollary \ref{corollarymain}.

Let $L=3$,
\begin{eqnarray}\Sigma_{Y_\mathcal{L}} = \left[
\begin{array}{rrrrr}
    1.0000  &  0.9500  &  0.7000\\
    0.9500  &  1.0000  &  0.6000\\
    0.7000  &  0.6000  &  1.0000
\end{array}\right],
\end{eqnarray}
and
\begin{eqnarray}
D_\mathcal{L}=(0.4, 0.45, 0.3)^T.
\end{eqnarray}

Let $\pi=\{1,2,3\}$ and
\begin{eqnarray}\Sigma_{N_\mathcal{L}} = \left[
\begin{array}{rrrrr}
    0.4827  &  0.5074  &       0\\
    0.5074  &  0.6205  &       0\\
         0  &       0  &  0.0512
\end{array}\right]
\end{eqnarray}
be a $\pi^{(K)}$ p.d. block diagonal matrix with $K=1$.

Then the BT minimum sum-rate bound for the MT source coding problem defined by $\Sigma_{Y_\mathcal{L}}$ and $D_\mathcal{L}$ is tight, since $\tilde{\boldsymbol{\Gamma}}$ and ${\boldsymbol{\Lambda}}$ defined in (\ref{Gammadef}) and (\ref{Lambdacorollary}) satisfy
\begin{eqnarray}
\mathrm{sign}([\tilde{\boldsymbol{\Gamma}}]_{1,2})\cdot[{\boldsymbol{\Lambda}}]_{1,2}&=&0.3596 \nonumber\\&\ge&0.2815~=~2[\tilde{\boldsymbol{\Gamma}}]_{1,2}-\sqrt{[\tilde{\boldsymbol{\Gamma}}]_{1,1}[\tilde{\boldsymbol{\Gamma}}]_{2,2}}.\nonumber
\end{eqnarray}

We have shown that the sum-rate tightness in the above numerical example is ensured by Corollary \ref{corollarymain}. In addition, it can be verified numerically that it does not satisfy the tightness condition provided by Wang {\em et al.} \cite{WangISIT09}.

\section{Conclusions}

We have provided a new sufficient condition for the BT sum-rate bound of quadratic Gaussian MT source coding problem to be tight. The matrix-distortion constrained two-terminal source coding problem was investigated with partially tighter sum-rate bound given. This result was then used to derive a new lower bound on the sum-rate of quadratic Gaussian $L$-terminal source coding problem, which was shown to coincide with the BT sum-rate bound under certain subgradient-based KKT conditions. To highlight the superior inclusiveness of our new condition, examples were shown to satisfy the tightness condition derived in this paper (while excluded from the so far best known tightness condition given by Wang {\em et al.} \cite{WangISIT09}).

We are currently investigating possible generalizations of techniques used
in the current paper to allow $3\times 3$ (or even larger) blocks in the observation noise covariance matrix.

\vspace{0.1in}
\begin{center}
{\bf Appendix A:} Proof of Lemma \ref{lemma2term}
\end{center}
\vspace{0.1in}
\begin{proof}
Before proving Lemma \ref{lemma2term}, we define an equivalent two-terminal problem, with
\begin{eqnarray}\Sigma_{Y_\mathcal{L}}=\left[\begin{array}{cc}v_1^2&\rho v_1v_2\\\rho v_1v_2& v_2^2\end{array}\right],\quad\mathrm{~and~} \boldsymbol{D}=\left[\begin{array}{cc}1&\theta\\\theta&1\end{array}\right].
\end{eqnarray}
Then we need to prove $R_{sum}(\Sigma_{Y_\mathcal{L}},\boldsymbol{D})\ge R_{\mu}(\Sigma_{Y_\mathcal{L}},\boldsymbol{D})$ and $R_{sum}(\Sigma_{Y_\mathcal{L}},\boldsymbol{D})\ge R_{lb}(\Sigma_{Y_\mathcal{L}},\boldsymbol{D})$.

To prove Lemma \ref{lemma2term}, let
\begin{eqnarray}
X=Y_1+Y_2+Z,
\end{eqnarray}
where $Z\sim \mathcal{N}(0,\sigma_{Z}^2)$ with $\sigma_{Z}^2=\frac{v_1v_2(1-\rho^2)}{\rho}$. Then the variance of $X$ can be computed as $\sigma_X^2=\frac{(v_1+v_2\rho)(v_2+v_1\rho)}{\rho}$, and it can be easily verified that
{\begin{eqnarray}
Y_\mathcal{L}&=&[\alpha_1, \alpha_2]^T\cdot X + [\tilde{N}_1, \tilde{N}_2]^T,
\end{eqnarray}}
with $\alpha_1=\frac{v_1\rho}{v_2+v_1\rho}$, $\alpha_2=\frac{v_2\rho}{v_1+v_2\rho}$, $[\tilde{N}_1, \tilde{N}_2]^T\sim \mathcal{N}\Big(\boldsymbol{0},\mathrm{diag}(\sigma_{\tilde{N}_1}^2,\sigma_{\tilde{N}_2}^2)\Big)$, and $\sigma_{\tilde{N}_1}^2=\frac{v_1^2v_2(1-\rho^2)}{v_2+v_1\rho}$, $\sigma_{\tilde{N}_2}^2= \frac{v_2^2v_1(1-\rho^2)}{v_1+v_2\rho}$. Moreover, any scheme that achieves a distortion matrix $\boldsymbol{D}$ on $Y_\mathcal{L}$ must be able to achieve a distortion of $[1~1]\cdot\boldsymbol{D}\cdot[1~1]^T+\sigma_{Z}^2$ on $X$.

Hence
{\begin{eqnarray}
H(W_\mathcal{L})&=&I(\boldsymbol{Y}_\mathcal{L},\boldsymbol{X};W_\mathcal{L})\nonumber\\
&=&I(\boldsymbol{X};W_\mathcal{L})+\sum_{i=1}^2I(\boldsymbol{Y}_i;W_i|\boldsymbol{X})\label{Lemma2proof1}\\
&=&h(\boldsymbol{X})-h(\boldsymbol{X}|W_\mathcal{L})+\sum_{i=1}^2h(\boldsymbol{Y}_i|\boldsymbol{X})-h(\boldsymbol{Y}_i;W_i|\boldsymbol{X})\nonumber\\
&\ge&\frac{n}{2}\log\frac{\sigma_{X}^2}{[1~1]\cdot\boldsymbol{D}\cdot[1~1]^T+\sigma_{Z}^2}
+\frac{n}{2}\log \frac{\sigma_{\tilde{N}_1}^2\sigma_{\tilde{N}_2}^2}{\gamma_1\gamma_2}\label{proofLemma102}\\
&\ge&\frac{n}{2}\log\frac{\sigma_{X}^2}{2+2\theta+\frac{v_1v_2(1-\rho^2)}{\rho}}
+\frac{n}{2}\log \frac{v_1^3v_2^3(1-\rho^2)^2}{(v_2+v_1\rho)(v_1+v_2\rho)\gamma_1\gamma_2},\nonumber
\end{eqnarray}}
where (\ref{Lemma2proof1}) uses the fact that $W_i\rightarrow Y_i^n\rightarrow\boldsymbol{X}\rightarrow(Y_j^n,W_j)$ form a Markov chain for any $i,j\in\{1,2\}$ and $i\neq j$, in (\ref{proofLemma102}) we define $\gamma_i\stackrel{\Delta}{=}\frac{1}{n}\sum_{j=1}^n\mathrm{var}(Y_{i,j}|W_i,\boldsymbol{X})$ and use the fact that Gaussian random variables maximize entropy over those with a fixed variance.

On the other hand, due to \cite[Lemma 1]{WangISIT09}, we known that $\frac{1}{n}\sum_{i=1}^n\mathrm{cov}(Y_\mathcal{L,i}|X_i,W_\mathcal{L})=\mathrm{diag}(\gamma_1,\gamma_2)$. Then \cite[Lemma 3]{WangISIT09} implies that
\begin{eqnarray}
\frac{1}{n}\sum_{i=1}^n\mathrm{cov}(Y_\mathcal{L,i}|X_i,W_\mathcal{L})\preceq \left(\boldsymbol{D}^{-1}+\Sigma_{\tilde{N}_\mathcal{L}}^{-1}-\Sigma_{Y_\mathcal{L}}^{-1}\right)^{-1}, \mathrm{~with~}\Sigma_{\tilde{N}_\mathcal{L}}=\big(\mathrm{diag}(\sigma_{\tilde{N}_1}^2,\sigma_{\tilde{N}_2}^2)\big),
\end{eqnarray}
i.e.,
\begin{eqnarray}
\mathrm{diag}(\gamma_1,\gamma_2)\preceq\left(\left[\begin{array}{cc}1&\theta\\\theta&1\end{array}\right]^{-1}+\frac{\rho}{v_1v_2(1-\rho^2)}\left[\begin{array}{cc}1&1\\1&1\end{array}\right]\right)^{-1},\label{proofLemma103}
\end{eqnarray}
which can be combined with (\ref{proofLemma102}) to form a semi-definite optimization problem that minimizes
\begin{eqnarray}
\mathcal{F}(\gamma_1,\gamma_2)&\stackrel{\Delta}{=}&\frac{1}{2}\log \frac{1}{\gamma_1\gamma_2}\end{eqnarray} over $\gamma_1$ and $\gamma_2$ subject to
\begin{eqnarray}
\mathcal{G}(\gamma_1,\gamma_2)&\stackrel{\Delta}{=}&\left[\begin{array}{cc}1&\theta\\\theta&1\end{array}\right]^{-1} +\frac{\rho}{v_1v_2(1-\rho^2)}\left[\begin{array}{cc}1&1\\1&1\end{array}\right]-\mathrm{diag}(\gamma_1^{-1},\gamma_2^{-1})\preceq\boldsymbol{0}.
\end{eqnarray}
The Lagrangian is
\begin{eqnarray}
\mathbbm{L}(\gamma_1,\gamma_2)&=&\mathcal{F}(\gamma_1,\gamma_2)+\mathrm{tr}\Big(\Lambda\mathcal{G}(\gamma_1,\gamma_2)\Big),
\end{eqnarray}
where $\Lambda$ is a positive semi-definite matrix. Then the KKT condition is given by
\begin{eqnarray}
\nabla_{\gamma_i}\mathbbm{L}(\gamma_1,\gamma_2)&=&0,\quad i=1,2,\label{KKT201}\\
\mathcal{G}(\gamma_1,\gamma_2)&\preceq&\boldsymbol{0},\label{KKT202}\\
\Lambda\mathcal{G}(\gamma_1,\gamma_2)&=&\boldsymbol{0}.\label{KKT203}
\end{eqnarray}
Solving the (\ref{KKT201}) and (\ref{KKT203}), we get two sets of solutions, namely,
\begin{eqnarray}
\gamma_1&=&1-\theta,\nonumber\\
\gamma_2&=&1-\theta,\nonumber\\
\Lambda&=&\frac{1-\theta}{2}\cdot\left[\begin{array}{cc}1&-1\\-1&1\end{array}\right],
\end{eqnarray}
and
\begin{eqnarray}
\gamma_1&=&\frac{v_1v_2(1-\rho^2)(1+\theta)}{v_1v_2(1-\rho^2)+2\rho(1+\theta)},\nonumber\\
\gamma_2&=&\frac{v_1v_2(1-\rho^2)(1+\theta)}{v_1v_2(1-\rho^2)+2\rho(1+\theta)},\nonumber\\
\Lambda&=&\frac{v_1v_2(1-\rho^2)(1+\theta)}{v_1v_2(1-\rho^2)+2\rho(1+\theta)}\cdot\left[\begin{array}{cc}1&1\\1&1\end{array}\right].
\end{eqnarray}
Then it is easy to verify that the first set of solution satisfies (\ref{KKT202}) if $\theta\ge\tilde{\theta}$, while the second set of solution satisfies (\ref{KKT202}) if $\theta\le\tilde{\theta}$. Hence the optimal solutions of $\gamma_1$ and $\gamma_2$ are
\begin{eqnarray}
\gamma_1=\gamma_2&=&\left\{\begin{array}{cc}\frac{v_1v_2(1-\rho^2)(1+\theta)}{v_1v_2(1-\rho^2)+2\rho(1+\theta)}&\theta\le\tilde{\theta}\\1-\theta&\theta>\tilde{\theta}\end{array}\right.,
\end{eqnarray}
which directly lead to (\ref{Lemma1statement}).

To prove tightness of the lower bound $\underline{R}_{sum}(\Sigma_{Y_\mathcal{L}},\boldsymbol{D})$ when $\theta\le\tilde{\theta}$, we construct a BT scheme with distortion matrix
\begin{eqnarray}
\tilde{\boldsymbol{D}}&=&(\Sigma_{Y_\mathcal{L}}^{-1}+\mathrm{diag}(q_1,q_2)^{-1})^{-1}\nonumber\\
&=&\left[\begin{array}{cc}\frac{(1+\theta)(v_1v_2(1-\rho^2)+\rho(1+\theta))}{(v_1v_2(1-\rho^2)+2\rho(1+\theta))}&\frac{\rho(1+\theta)^2}{(v_1v_2(1-\rho^2)+2\rho(1+\theta))}\\ \frac{\rho(1+\theta)^2}{(v_1v_2(1-\rho^2)+2\rho(1+\theta))}&\frac{(1+\theta)(v_1v_2(1-\rho^2)+\rho(1+\theta))}{(v_1v_2(1-\rho^2)+2\rho(1+\theta))}\end{array}\right],\nonumber
\end{eqnarray}
and sum-rate
\begin{eqnarray}
\frac{1}{2}\log\frac{|\Sigma_{Y_\mathcal{L}}|}{|\tilde{\boldsymbol{D}}|}=\frac{1}{2}\log\frac{v_1^2v_2^2(1-\rho^2)}{\frac{v_1v_2(1+\theta)^2(1-\rho^2)}{(v_1v_2(1-\rho^2)+2\rho(1+\theta))}}= \underline{R}_{sum}(\Sigma_{Y_\mathcal{L}},\boldsymbol{D}),
\end{eqnarray}
where
\begin{eqnarray}
q_1&=&\frac{v_1^2v_2(1-\rho^2)(1+\theta)}{v_1^2v_2(1-\rho^2)-(v_2-\rho v_1)(1+\theta)},\nonumber\\
q_2&=&\frac{v_1v_2^2(1-\rho^2)(1+\theta)}{v_1v_2^2(1-\rho^2)-(v_1-\rho v_2)(1+\theta)}.\nonumber
\end{eqnarray}
Then tightness is proved by verifying 
\begin{eqnarray}
{\boldsymbol{D}}-\tilde{\boldsymbol{D}}&=&\frac{\rho(1-\theta^2)-v_1v_2\theta(1-\rho^2)}{(v_1v_2(1-\rho^2)+2\rho(1+\theta))}\cdot\left[\begin{array}{cc}1&-1\\-1&1\end{array}\right] \nonumber\\&\succeq&\boldsymbol{0},
\end{eqnarray}
where the last matrix inequality is due to the facts that $f_1(\theta)\stackrel{\Delta}{=}(v_1v_2(1-\rho^2)+2\rho(1+\theta))>0$, $f_2(\theta)\stackrel{\Delta}{=}\rho(1-\theta^2)-v_1v_2\theta(1-\rho^2)$ is monotone decreasing in the range $\theta\in[-1,\tilde{\theta})$, $f_2(\tilde{\theta})=0$, and the assumption that $\theta\le\tilde{\theta}$.
\end{proof}

\begin{center}
{\bf Appendix B:} Proof of Lemma \ref{lemmasubdiff}
\end{center}

\begin{proof}
First, due to the assumption that $\tilde{\boldsymbol{D}}^{-1}-\Sigma_{Y_\mathcal{L}}^{-1}$ is a p.s.d. diagonal matrix, we must have
\begin{eqnarray}
\theta&=&\left\{\begin{array}{cc}\frac{\sqrt{1-2\rho^2+\rho^4+4\rho^2d_1^2d_2^2}-(1-\rho^2)}{2\rho d_1d_2}&\rho\ge 0\\ \frac{-\sqrt{1-2\rho^2+\rho^4+4\rho^2d_1^2d_2^2}-(1-\rho^2)}{2\rho d_1d_2}&\rho<0\end{array}\right.,
\end{eqnarray}
with $d_1=\sqrt{D_1}$ and $d_2=\sqrt{D_2}$. Now since
\begin{eqnarray}
\underline{R}_{sum}(\Sigma_{Y_\mathcal{L}},{\boldsymbol{D}})&=&\max\Big\{R_{lb}(\Sigma_{Y_\mathcal{L}},{\boldsymbol{D}}), R_{\mu}(\Sigma_{Y_\mathcal{L}},{\boldsymbol{D}})\Big\},\nonumber
\end{eqnarray}
we compute
\begin{eqnarray}
\nabla_{{\boldsymbol{D}}} R_{lb}(\Sigma_{Y_\mathcal{L}},{\boldsymbol{D}})\mid_{\boldsymbol{D}=\tilde{\boldsymbol{D}}} &=&\kappa\cdot\left[\begin{array}{cc}\frac{1}{D_1}&\frac{s(1-2|\theta|)}{\sqrt{D_1D_2}}\\\frac{s(1-2|\theta|)}{\sqrt{D_1D_2}}&\frac{1}{D_2}\end{array}\right],\nonumber\\
\nabla_{{\boldsymbol{D}}} R_{\mu}(\Sigma_{Y_\mathcal{L}},{\boldsymbol{D}})\mid_{\boldsymbol{D}=\tilde{\boldsymbol{D}}} &=&\chi\cdot\left[\begin{array}{cc}\frac{1}{D_1}&\frac{s}{\sqrt{D_1D_2}}\\\frac{s}{\sqrt{D_1D_2}}&\frac{1}{D_2}\end{array}\right],
\end{eqnarray}
where
\begin{eqnarray}
\kappa&=&\frac{\rho^4-2d_1d_2\rho^3-2\rho^2+4\rho^2d_1^2d_2^2+2d_1d_2\rho+1}{2(1-\rho^2)^2}-\frac{\rho^2+2d_1d_2\rho-1}{2(1-\rho^2)^2}\sqrt{1-2\rho^2+\rho^4+4\rho^2d_1^2d_2^2},\nonumber\\
\chi&=&-\frac{\rho^4+2d_2\rho^3d_1+4\rho^2d_2^2d_1^2-2\rho^2-2d_2\rho d_1+1}{2(1-\rho^2)^2}+\frac{2d_1d_2\rho -1+\rho^2}{2(1-\rho^2)^2}\sqrt{1-2\rho^2+\rho^4+4\rho^2d_1^2d_2^2}.
\end{eqnarray}
Finally, it is easy to verify that
\begin{eqnarray}
-\tilde{\boldsymbol{D}}\cdot\nabla_{{\boldsymbol{D}}} R_{lb}(\Sigma_{Y_\mathcal{L}},{\boldsymbol{D}})\mid_{\boldsymbol{D}=\tilde{\boldsymbol{D}}}\cdot\tilde{\boldsymbol{D}}&=&\left[\begin{array}{cc}{D}_1&s(1-2|\theta|)\sqrt{{D}_1{D}_2}\\ s(1-2|\theta|)\sqrt{{D}_1{D}_2}&{D}_2\end{array}\right],\nonumber\\
-\tilde{\boldsymbol{D}}\cdot\nabla_{{\boldsymbol{D}}} R_{\mu}(\Sigma_{Y_\mathcal{L}},{\boldsymbol{D}})\mid_{\boldsymbol{D}=\tilde{\boldsymbol{D}}}\cdot\tilde{\boldsymbol{D}}&=&\left[\begin{array}{cc}{D}_1&s\sqrt{{D}_1{D}_2}\\ s\sqrt{{D}_1{D}_2}&{D}_2\end{array}\right],\nonumber
\end{eqnarray}
and Lemma \ref{lemmasubdiff} readily follows.
\end{proof}

\begin{center}
{\bf Appendix C:} Proof of Lemma \ref{lemmasamestruct}
\end{center}

\begin{proof}
To prove Lemma \ref{lemmasamestruct}, we need to use \cite[Lemma 1]{WangISIT09}, which is stated in the following proposition for the sake of completion.

\begin{proposition}\label{proplemma1}
For integers $n$, $m$ and random variables $X$ and $\omega$, let $\boldsymbol{X}$ be a row vector of $n$ independent drawings of $X$, and $\boldsymbol{Y}(\omega)$ be any $1\times m$ vector of measurable functions of $\omega$. Then it holds that
\begin{eqnarray}
E\left[\Big(\boldsymbol{X}-E(\boldsymbol{X}|\omega)\Big)^T\boldsymbol{Y}(\omega)\right]&=&\boldsymbol{0}_{n\times m}.
\end{eqnarray}
\end{proposition}

Now (\ref{lemmasamestruct1}) and the definition of $W_\mathcal{L}$ imply that the Markov chains $W_i\rightarrow\boldsymbol{Y}_{i}\rightarrow\boldsymbol{X}_\mathcal{M}\rightarrow(\boldsymbol{Y}_{j},W_{j})$ and $W_j\rightarrow\boldsymbol{Y}_{j}\rightarrow\boldsymbol{X}_\mathcal{M}\rightarrow(\boldsymbol{Y}_{i},W_{i})$ hold. Hence (\ref{lemmasamestruct2}) must hold since
\begin{eqnarray}
\hspace{-0.15in}\Big[\mathrm{cov}(\boldsymbol{Y}_\mathcal{L}|\boldsymbol{X}_\mathcal{M},W_\mathcal{L})\Big]_{i,j} &=&E\Big[\big(\boldsymbol{Y}_i-E(\boldsymbol{Y}_i|\boldsymbol{X}_\mathcal{M},W_\mathcal{L})\big) \big(\boldsymbol{Y}_j-E(\boldsymbol{Y}_j|\boldsymbol{X}_\mathcal{M},W_\mathcal{L})\big)^T\Big]\nonumber\\
&=&E\Big[\big(\boldsymbol{Y}_i-E(\boldsymbol{Y}_i|\boldsymbol{X}_\mathcal{M},W_{i})\big) \big(\boldsymbol{Y}_j-E(\boldsymbol{Y}_j|\boldsymbol{X}_\mathcal{M},W_{j})\big)^T\Big]\label{newlemma4proof1}\\
&=&E\Big[\big(\boldsymbol{Y}_i-E(\boldsymbol{Y}_i|\boldsymbol{X}_\mathcal{M},\boldsymbol{Y}_j,W_{i},W_{j})\big) \big(\boldsymbol{Y}_j-E(\boldsymbol{Y}_j|\boldsymbol{X}_\mathcal{M},W_{j})\big)^T\Big]\label{newlemma4proof2}\\
&=&0,\label{newlemma4proof3}
\end{eqnarray}
where (\ref{newlemma4proof1}) and (\ref{newlemma4proof2}) are due to the above two Markov chains, and (\ref{newlemma4proof3}) used Proposition \ref{proplemma1} and the fact that $\Big(\boldsymbol{Y}_j-E(\boldsymbol{Y}_j|\boldsymbol{X}_\mathcal{M},W_{j})\Big)$ is a function of $\omega\stackrel{\Delta}{=}(\boldsymbol{X}_\mathcal{M}, \boldsymbol{Y}_j, W_{i}, W_{j})$.
\end{proof}

\begin{center}
{\bf Appendix D:} Proof of Lemma \ref{lemmaconnect}
\end{center}

\begin{proof}
First, given $\Sigma_{N_\mathcal{L}}\in\Upsilon_K(\pi)$ and $\Sigma_{N_\mathcal{L}}\preceq\Sigma_{Y_\mathcal{L}}$, we can always apply (\ref{Xconst1}) to find an $M\times L$ matrix $\boldsymbol{H}$ and (\ref{Xconst2}) to construct $M$ remote sources $X_\mathcal{M}$ such that (\ref{correlationstruct}) holds. This implies that $\Sigma_{N_\mathcal{L}}=\mathrm{cov}({\boldsymbol{Y}_\mathcal{L}|\boldsymbol{X}_\mathcal{M}})\in\Upsilon_K(\pi)$.
Then we can apply Lemma \ref{lemmasamestruct}, and obtain that $\mathrm{cov}({\boldsymbol{Y}_\mathcal{L}|\boldsymbol{X}_\mathcal{M},W_\mathcal{L}})\in\Upsilon_K(\pi)$. Hence we can denote
\begin{eqnarray}
\boldsymbol{\Gamma}~\stackrel{\Delta}{=}~\mathrm{cov}({\boldsymbol{Y}_\mathcal{L}|\boldsymbol{X}_\mathcal{M},W_\mathcal{L}})~\in~\Upsilon_K(\pi),
\end{eqnarray}
which takes form of (\ref{Gammadef}).

On the other hand, due to (\ref{correlationstruct}), we know that any scheme that achieves a distortion matrix of $\boldsymbol{D}$ on $Y_\mathcal{L}$ must be able to achieve a distortion matrix of $\boldsymbol{A}\boldsymbol{D}\boldsymbol{A}^T+\boldsymbol{B}$ on $X_\mathcal{M}$.

Similar to (\ref{Lemma2proof1}), we write
\begin{eqnarray}
&&H(W_\mathcal{L})\nonumber\\&=&I(\boldsymbol{Y}_\mathcal{L},\boldsymbol{X};W_\mathcal{L})\nonumber\\
&=&I(\boldsymbol{X};W_\mathcal{L})+\sum_{i=1}^KI(\boldsymbol{Y}_{\{\pi_{2k-1},\pi_{2k}\}};W_{\{\pi_{2k-1},\pi_{2k}\}}|\boldsymbol{X}_\mathcal{M}) +\sum_{i=K+1}^LI(\boldsymbol{Y}_{\pi_{i}};W_{\pi_{i}}|\boldsymbol{X}_\mathcal{M})\label{proofLemma201}\\
&=&h(\boldsymbol{X})-h(\boldsymbol{X}|W_\mathcal{L})+\sum_{i=1}^KI(\boldsymbol{Y}_{\{\pi_{2k-1},\pi_{2k}\}};W_{\{\pi_{2k-1},\pi_{2k}\}}|\boldsymbol{X}_\mathcal{M}) \nonumber\\&&+\sum_{i=K+1}^L\big(h(\boldsymbol{Y}_{\pi_{i}}|\boldsymbol{X}_\mathcal{M})-h(\boldsymbol{Y}_{\pi_{i}};W_{\pi_{i}}|\boldsymbol{X}_\mathcal{M})\big)\\
&\ge&\frac{1}{2}\log\frac{|\Sigma_{X_\mathcal{M}}|}{|\boldsymbol{A}\boldsymbol{D}\boldsymbol{A}^T+\boldsymbol{B}|}
+\sum_{i=1}^KI(\boldsymbol{Y}_{\{\pi_{2k-1},\pi_{2k}\}};W_{\{\pi_{2k-1},\pi_{2k}\}}|\boldsymbol{X}_\mathcal{M}) +\frac{1}{2}\sum_{i=K+1}^L\log \frac{\sigma_{N_{\pi_i}}^2}{\gamma_i},\label{proofLemma202}
\end{eqnarray}
where (\ref{proofLemma202}) comes from the assumption that the achieved distortion is no larger than $\boldsymbol{D}$ in the positive definite sense, and the definitions $\mathrm{cov}({\boldsymbol{Y}_{\{\pi_{2k-1},\pi_{2k}\}}|W_{\{\pi_{2k-1},\pi_{2k}\}},\boldsymbol{X}_\mathcal{M}})=\boldsymbol{\Gamma}_k$ and $\gamma_i=\frac{1}{n}\sum_{j=1}^n\mathrm{var}(Y_{i,j}|W_i,\boldsymbol{X})$. Now comparing (\ref{lemmaconnentstatement}) with (\ref{proofLemma202}), we only need to show that
\begin{eqnarray}
I(\boldsymbol{Y}_{\{\pi_{2k-1},\pi_{2k}\}};W_{\{\pi_{2k-1},\pi_{2k}\}}|\boldsymbol{X}_\mathcal{M}) &\ge&nR_{sum}(\Sigma_{\boldsymbol{Y}_{\{\pi_{2k-1},\pi_{2k}\}}|\boldsymbol{X}_\mathcal{M}},\boldsymbol{\Gamma}_k)\label{lemma4proof1}
\end{eqnarray}
holds for any $k\in\mathcal{K}$.

Assume that (\ref{lemma4proof1}) does not hold for some $k\in\mathcal{K}$, i.e., there exist encoders $\psi_{\pi_{2k-1}}^{(n)}$ and $\psi_{\pi_{2k}}^{(n)}$ such that
\begin{eqnarray}
\mathrm{cov}({\boldsymbol{Y}_{\{\pi_{2k-1},\pi_{2k}\}}|W_{\{\pi_{2k-1},\pi_{2k}\}},\boldsymbol{X}_\mathcal{M}})&=&\boldsymbol{\Gamma}_k,\nonumber\\
I(\boldsymbol{Y}_{\{\pi_{2k-1},\pi_{2k}\}};W_{\{\pi_{2k-1},\pi_{2k}\}}|\boldsymbol{X}_\mathcal{M}) &<&nR_{sum}(\Sigma_{\boldsymbol{Y}_{\{\pi_{2k-1},\pi_{2k}\}}|\boldsymbol{X}_\mathcal{M}},\boldsymbol{\Gamma}_k).
\end{eqnarray}
Then consider the matrix-distortion constrained two-terminal problem with sources
\begin{eqnarray}
\tilde{Y}_{\{\pi_{2k-1},\pi_{2k}\}}\sim\mathcal{N}(\boldsymbol{0}, \Sigma_{\boldsymbol{Y}_{\{\pi_{2k-1},\pi_{2k}\}}|\boldsymbol{X}_\mathcal{M}})
\end{eqnarray}
and target distortion matrix $\boldsymbol{\Gamma}_k$. Now let $\boldsymbol{X}_\mathcal{M}$ be a length-$n$ block of samples independently draw from $X_\mathcal{M}=\boldsymbol{A} Y_\mathcal{L} + Z_\mathcal{L}$ according to (\ref{Xconst2}). Also assume that $\boldsymbol{X}_\mathcal{M}$ is independent of the sources $\tilde{\boldsymbol{Y}}_{\{\pi_{2k-1},\pi_{2k}\}}$ and available at both the encoders and the decoder. Let
\begin{eqnarray}
\bar{\boldsymbol{Y}}_{\{\pi_{2k-1},\pi_{2k}\}}&=&\tilde{\boldsymbol{Y}}_{\{\pi_{2k-1},\pi_{2k}\}}+\boldsymbol{H}_{\mathcal{M},\{\pi_{2k-1},\pi_{2k}\}}^T\boldsymbol{X}_\mathcal{M},
\end{eqnarray}
where $\boldsymbol{H}$ is the $M\times L$ matrix satisfying (\ref{correlationstruct}).
It is obvious that $\bar{\boldsymbol{Y}}_{\{\pi_{2k-1},\pi_{2k}\}}$ has a covariance matrix of $\Sigma_{\boldsymbol{Y}_{\{\pi_{2k-1},\pi_{2k}\}}}$, hence we can blindly apply the same encoders $\psi_{\pi_{2k-1}}^{(n)}$ and $\psi_{\pi_{2k}}^{(n)}$ on $\bar{\boldsymbol{Y}}_{\{\pi_{2k-1},\pi_{2k}\}}$ to generate $W_{\{\pi_{2k-1},\pi_{2k}\}}$ before using Slepian-Wolf coding with decoder side information $\boldsymbol{X}_\mathcal{M}$, to achieve a final rate of
\begin{eqnarray}
H(W_{\{\pi_{2k-1},\pi_{2k}\}}|\boldsymbol{X}_\mathcal{M})
&=&I(\boldsymbol{Y}_{\{\pi_{2k-1},\pi_{2k}\}};W_{\{\pi_{2k-1},\pi_{2k}\}}|\boldsymbol{X}_\mathcal{M})\nonumber\\
&<&nR_{sum}(\Sigma_{\boldsymbol{Y}_{\{\pi_{2k-1},\pi_{2k}\}}|\boldsymbol{X}_\mathcal{M}},\boldsymbol{\Gamma}_k),
\end{eqnarray}
and a distortion matrix of $\boldsymbol{\Gamma}_k=\mathrm{cov}({\boldsymbol{Y}_{\{\pi_{2k-1},\pi_{2k}\}}|W_{\{\pi_{2k-1},\pi_{2k}\}},\boldsymbol{X}_\mathcal{M}})$, which contradicts with the definition of $R_{sum}(\Sigma_{\boldsymbol{Y}_{\{\pi_{2k-1},\pi_{2k}\}}|\boldsymbol{X}_\mathcal{M}},\boldsymbol{\Gamma}_k)$. Then Lemma \ref{lemmaconnect} follows from (\ref{proofLemma201}), (\ref{lemma4proof1}), and Lemma \ref{lemma2term}.
\end{proof}

\begin{center}
{\bf Appendix E:} Proof of Lemma \ref{lemmaBD}
\end{center}
\begin{proof}
First, it is obvious that
$R^{BT}_{sum}(\Sigma_{{Y}_\mathcal{L}},{D}_\mathcal{L})
~=~R_{sum}^{BT}(\Sigma_{\bar{Y}_\mathcal{K}},\bar{D}_\mathcal{K})$.
Then assume that there is a sequence of schemes
$\{(\phi^{(n)}_\mathcal{L},\psi^{(n)}_\mathcal{L}):n\in\mathbb{N}^+\}$ such that
\begin{eqnarray}
\limsup_{n\rightarrow\infty} \sum_{j\in\mathcal{L}} R_j^{(n)}
&<&R_{sum}^{BT}(\Sigma_{\bar{Y}_\mathcal{K}},\bar{D}_\mathcal{K}),\\
\limsup_{n\rightarrow\infty} \frac{1}{n}E\Big[({Y}_{j,i}-\hat{{Y}}_{j,i})^2
\Big]&\le&D_j,\mathrm{~~for~any~}j\in\mathcal{L}.
\end{eqnarray}
Now consider another sequence of schemes
$\{(\bar{\phi}^{(n)}_\mathcal{L},\bar{\psi}^{(n)}_\mathcal{L}):n\in
\mathbb{N}^+\}$ such that for any $k\in\mathcal{K}$,
\begin{eqnarray}
\bar{\phi}^{(n)}_{\mathbbm{i}(k)}(\boldsymbol{Y}_{\mathbbm{i}(k)})
&=&\boxtimes_{j\in\mathcal{S}_k}\bar{W}_j,\\
\bar{\phi}^{(n)}_j(\boldsymbol{Y}_j)&\equiv&0\mathrm{~for~any~}j
\in\mathcal{S}_k-\{\mathbbm{i}(k)\},
\end{eqnarray}
where
\begin{eqnarray}
\bar{W}_{\mathbbm{i}(k)}&\stackrel{\Delta}{=}&{W}_{\mathbbm{i}(k)}
~=~{\phi}^{(n)}_{\mathbbm{i}(k)}(\boldsymbol{Y}_{\mathbbm{i}(k)}),\\
\bar{W}_{j}&\stackrel{\Delta}{=}
&{\phi}^{(n)}_{j}(\boldsymbol{Y}_{\mathbbm{i}(k)}+\boldsymbol{Z}_j),
\end{eqnarray}
with $\bar{Z}_j\sim\mathcal{N}(0,\sigma_{Z_j}^2)$ being independent of $Y_\mathcal{L}$, ``$\boxtimes$" denotes Cartesian product, and
\begin{eqnarray}
\bar{\psi}^{(n)}_{\mathbbm{i}(k)}(W_\mathcal{L})
&=&{\psi}^{(n)}_{\mathbbm{i}(k)}(\bar{W}_\mathcal{L}).
\end{eqnarray}

Then we must have
\begin{eqnarray}
R_{sum}(\phi^{(n)}_\mathcal{L},\psi^{(n)}_\mathcal{L})
&=&R_{sum}(\bar{\phi}^{(n)}_\mathcal{L},\bar{\psi}^{(n)}_\mathcal{L}),\\
\Rightarrow \limsup_{n\rightarrow\infty}R_{sum}(\bar{\phi}^{(n)}_\mathcal{L},
\bar{\psi}^{(n)}_\mathcal{L})
&=&\limsup_{n\rightarrow\infty}R_{sum}(\phi^{(n)}_\mathcal{L},
\psi^{(n)}_\mathcal{L})\nonumber\\
&<&R_{sum}^{BT}(\Sigma_{\bar{Y}_\mathcal{K}},\bar{D}_\mathcal{K}),
\end{eqnarray}
and
\begin{eqnarray}
\limsup_{n\rightarrow\infty} \frac{1}{n}E\Big[\Big({Y}_{j,i}
-E({{Y}}_{j,i}|\bar{W}_\mathcal{L})\Big)^2\Big]
\hspace{-0.05in}&\le&\left\{\begin{array}{lr}D_j&j
=\mathbbm{i}(k)\mathrm{~for~some~}k\in\mathcal{K}\\
D_{\mathbbm{i}(k)}+\sigma_{Z_j}^2~\le~D_j&j\in\mathcal{S}_k
-\{\mathbbm{i}(k)\}\mathrm{~for~some~}k\in\mathcal{K}\end{array}
\right..\nonumber
\end{eqnarray}
Hence the sequence of schemes
$\{(\bar{\phi}^{(n)}_\mathcal{L},\bar{\psi}^{(n)}_\mathcal{L}):n
\in\mathbb{N}^+\}$ achieves the distortion vector $D_\mathcal{L}$
and a sum-rate smaller than
$R_{sum}^{BT}(\Sigma_{\bar{Y}_\mathcal{K}},\bar{D}_\mathcal{K})$.
On the other hand,
$\{(\bar{\phi}^{(n)}_\mathcal{L},\bar{\psi}^{(n)}_\mathcal{L}):n
\in\mathbb{N}^+\}$ is also an achievable sequence of schemes for the
induced $K$-terminal problem, for which the BT sum-rate bound
$R^{BT}_{sum}(\Sigma_{{Y}_\mathcal{L}},{D}_\mathcal{L})
~=~R_{sum}^{BT}(\Sigma_{\bar{Y}_\mathcal{K}},\bar{D}_\mathcal{K})$
is known to be tight, leading to a contradiction.
\end{proof}

\begin{center}
{\bf Appendix F:} Proof of Lemma \ref{lemmaWangconditionBC}
\end{center}
\begin{proof}
We only need to show that if \begin{eqnarray}
\mathrm{diag}\Big((\tilde{\boldsymbol{D}}\odot\tilde{\boldsymbol{D}})^{-1}D\boldsymbol{1}\Big)&\succeq& \tilde{\boldsymbol{D}}^{-1}-\tilde{\boldsymbol{D}}^{-1}(\tilde{\boldsymbol{D}}^{-1}+\Theta^{-1} -\Sigma_{Y_\mathcal{L}}^{-1})^{-1}\tilde{\boldsymbol{D}}^{-1}\label{WangcondBCtheta}
\end{eqnarray}
holds for some p.s.d. diagonal matrix $\Theta=\mathrm{diag}(\mu_1,\mu_2,...,\mu_L)$ such that
\begin{eqnarray}
\Sigma_{Y_\mathcal{L}}&\succeq&\Theta,\label{WangcondBCtheta1}
\end{eqnarray}
then (\ref{WangcondBC}) must also hold.

In fact, due to the symmetric properties of block-circulant matrices, it is easy to show that if both (\ref{WangcondBCtheta}) and (\ref{WangcondBCtheta1}) hold for $\Theta=\mathrm{diag}(\mu_1,\mu_2,...,\mu_L)$, then they must also hold for
\begin{eqnarray}
\Theta^\dagger_k&=&\mathrm{diag}\Big(\mu_{\varsigma(k,1)},\mu_{\varsigma(k,2)}, \mu_{\varsigma(k+1,1)},\mu_{\varsigma(k+1,2)},...,\mu_{\varsigma(k+m-1,1)},\mu_{\varsigma(k+m-1,2)}\Big),
\end{eqnarray}
for any $k\in\{0,1,...,m-1\}$, as well as
\begin{eqnarray}
\Theta^\ddagger_k&=&\mathrm{diag}\Big(\mu_{\varsigma(k,2)},\mu_{\varsigma(k,1)}, \mu_{\varsigma(k+1,2)},\mu_{\varsigma(k+1,1)},...,\mu_{\varsigma(k+m-1,2)},\mu_{\varsigma(k+m-1,1)}\Big),
\end{eqnarray}
where $\varsigma(j,i)\stackrel{\Delta}{=}2\cdot(j~\mathrm{mod}~m)+i$.
Hence (\ref{WangcondBC}) must be true since
\begin{eqnarray}
\mathrm{diag}\Big((\tilde{\boldsymbol{D}}\odot\tilde{\boldsymbol{D}})^{-1}D\boldsymbol{1}\Big)&\succeq& \frac{1}{L}\sum_{k=1}^m\left[\tilde{\boldsymbol{D}}^{-1}-\tilde{\boldsymbol{D}}^{-1}(\tilde{\boldsymbol{D}}^{-1}+(\Theta^\dagger_k)^{-1} -\Sigma_{Y_\mathcal{L}}^{-1})^{-1}\tilde{\boldsymbol{D}}^{-1}\right] \nonumber\\&&+ \frac{1}{L}\sum_{k=1}^m\left[\tilde{\boldsymbol{D}}^{-1}-\tilde{\boldsymbol{D}}^{-1}(\tilde{\boldsymbol{D}}^{-1}+(\Theta^\ddagger_k)^{-1} -\Sigma_{Y_\mathcal{L}}^{-1})^{-1}\tilde{\boldsymbol{D}}^{-1}\right]\nonumber\\
&\succeq&\tilde{\boldsymbol{D}}^{-1}-\tilde{\boldsymbol{D}}^{-1}(\tilde{\boldsymbol{D}}^{-1}+(\frac{1}{L}\sum_{k=1}^m\Theta^\dagger_k +\frac{1}{L}\sum_{k=1}^m\Theta^\ddagger_k)^{-1} -\Sigma_{Y_\mathcal{L}}^{-1})^{-1}\tilde{\boldsymbol{D}}^{-1}\label{prooflemma51}\\
&\succeq&\tilde{\boldsymbol{D}}^{-1}-\tilde{\boldsymbol{D}}^{-1}(\tilde{\boldsymbol{D}}^{-1}+\lambda_{min}^{-1}\boldsymbol{I}_L -\Sigma_{Y_\mathcal{L}}^{-1})^{-1}\tilde{\boldsymbol{D}}^{-1}\label{prooflemma52},
\end{eqnarray}
where (\ref{prooflemma51}) is due to the concavity of $\tilde{\boldsymbol{D}}^{-1}-\tilde{\boldsymbol{D}}^{-1}(\tilde{\boldsymbol{D}}^{-1}+\Theta^{-1} -\Sigma_{Y_\mathcal{L}}^{-1})^{-1}\tilde{\boldsymbol{D}}^{-1}$ with respect to $\Theta$, and (\ref{prooflemma52}) uses the fact that
\begin{eqnarray}
&&\Sigma_{Y_\mathcal{L}}~\succeq~\Theta^\dagger_k,~~\Sigma_{Y_\mathcal{L}}~\succeq~\Theta^\ddagger_k\nonumber\\
&\Rightarrow&\Sigma_{Y_\mathcal{L}}~\succeq~\frac{1}{L}\sum_{k=1}^m\Theta^\dagger_k +\frac{1}{L}\sum_{k=1}^m\Theta^\ddagger_k~=~\frac{1}{L}\sum_{i=1}^L \mu_i \boldsymbol{I}_L\nonumber\\&\Rightarrow&\frac{1}{L}\sum_{i=1}^L \mu_i\le\lambda_{min}.
\end{eqnarray}
\end{proof}

\end{document}